\pdfoutput=1
\RequirePackage{ifpdf}
\ifpdf 
\documentclass[pdftex]{sigma}
\else
\documentclass{sigma}
\fi

\usepackage{bbm}

\newcommand{\RP}{{\cal X}}
\newcommand{\Cov}{{\rm Cov}}
\newcommand{\Var}{{\rm Var}}
\newcommand{\Or}{\mathcal{O}}
\newcommand{\OR}{\Theta}
\newcommand{\Ai}{\mathrm{Ai}}
\newcommand{\Pb}{\mathbb{P}}
\newcommand{\E}{\mathbbm{E}}
\newcommand{\Id}{\mathbbm{1}}
\newcommand{\R}{\mathbb{R}}
\newcommand{\Z}{\mathbb{Z}}
\newcommand{\sgn}{\mathrm{sgn}}

\numberwithin{equation}{section}
\newtheorem{prop}{Proposition}[section]

\newtheorem{lem}[prop]{Lemma}
\newtheorem{cor}[prop]{Corollary}

{ \theoremstyle{definition}
\newtheorem{rem}[prop]{Remark}
\newtheorem{defin}[prop]{Definition}}

\begin{document}

\allowdisplaybreaks

\renewcommand{\thefootnote}{$\star$}

\newcommand{\arXivNumber}{1602.00486}

\renewcommand{\PaperNumber}{074}

\FirstPageHeading

\ShortArticleName{On Time Correlations for KPZ Growth in One Dimension}

\ArticleName{On Time Correlations for KPZ Growth\\ in One Dimension\footnote{This paper is a~contribution to the Special Issue
on Asymptotics and Universality in Random Matrices, Random Growth Processes, Integrable Systems and Statistical Physics in honor of Percy Deift and Craig Tracy. The full collection is available at \href{http://www.emis.de/journals/SIGMA/Deift-Tracy.html}{http://www.emis.de/journals/SIGMA/Deift-Tracy.html}}}

\Author{Patrik L.~FERRARI~$^\dag$ and Herbert SPOHN~$^\ddag$}

\AuthorNameForHeading{P.L.~Ferrari and H.~Spohn}

\Address{$^\dag$~Institute for Applied Mathematics, Bonn University,\\
\hphantom{$^\dag$}~Endenicher Allee 60, 53115 Bonn, Germany}
\EmailD{\href{mailto:ferrari@uni-bonn.de}{ferrari@uni-bonn.de}}

\Address{$^\ddag$~Zentrum Mathematik, TU M\"unchen, Boltzmannstrasse 3, D-85747 Garching, Germany}
\EmailD{\href{mailto:spohn@ma.tum.de}{spohn@ma.tum.de}}

\ArticleDates{Received March 17, 2016, in f\/inal form July 21, 2016; Published online July 26, 2016}

\Abstract{Time correlations for KPZ growth in $1+1$ dimensions are reconsidered. We discuss f\/lat, curved, and stationary initial conditions and are interested in the covariance of the height as a function of time at a f\/ixed point on the substrate. In each case the power laws of the covariance for short and long times are obtained. They are derived from a variational problem involving two independent Airy processes. For stationary initial conditions we derive an exact formula for the stationary covariance with two approaches: (1)~the variational problem and (2)~deriving the covariance of the time-integrated current at the origin for the corresponding driven lattice gas. In the stationary case we also derive the large time behavior for the covariance of the height gradients.}


\Classification{60K35; 82C22; 82B43}

\renewcommand{\thefootnote}{\arabic{footnote}}
\setcounter{footnote}{0}

\vspace{-2.5mm}

\section{Introduction}\label{sec1}
Because of novel experiments~\cite{Tak12,Tak13, TS12} and exact solutions (see surveys and lecture notes~\cite{BG12,Cor11,FS10,Qua11,QS15}), there is a continuing interest in growing surfaces in the Kardar--Parisi--Zhang (KPZ) universality class~\cite{KPZ86}, in particular for the case of $1+1$ dimensions. The object of interest is a height function $h(x,t)$ over the one-dimensional substrate space, $x \in \R$, at time $t \geq 0$, which evolves by a stochastic evolution. Examples are the KPZ equation itself, the single step model, polynuclear growth, Eden type growth, and more. The spatial statistics, $x \mapsto h(x,t)$ at large, but f\/ixed time $t$ is fairly well understood. The typical size of the height f\/luctuations is of order $t^{1/3}$ and the correlation length grows as $t^{2/3}$. The precise spatial statistics depends on the initial conditions. Three canonical cases have been singled out, which are f\/lat, step (also curved), and stationary. On the other hand, our understanding of the correlations in time is more fragmentary. For the point-to-point semi-discrete
directed polymer, which corresponds to curved initial data, Johansson~\cite{Joh15} recently derived the long time asymptotics of the joint
distribution of $(h(0,\tau t),h(0,t))$, $\tau$ f\/ixed, $t \to \infty$. In an earlier work on the same quantity~\cite{Dot13} Dotsenko obtains a replica solution of the KPZ equation. In both cases the f\/inal result is an inf\/inite series, from which it seems to be dif\/f\/icult to extract more explicit information\footnote{In~\cite{Dot16} progress has been achieved recently at the level of joint distribution functions for curved initial data in the limit $\tau\to 1$.}. For us, this state of af\/fairs is one motivation to reconsider the issue of the KPZ time correlations.

The most basic observable is the temporal correlation function
\begin{gather}
C^\diamond(t_0,t) = \Cov(h(0,t_0),h(0,t))
 =\E(h(0,t_0)h(0,t)) - \E(h(0,t_0))\E(h(0,t)).\label{1.1}
\end{gather}
Here the superscript ${}^\diamond$ stands for the initial conditions, which are denoted by either ``f\/lat'', ``step'', or ``stat''. In the stationary case the covariance depends only on $t- t_0$. But for f\/lat and curved both arguments have to be kept.

\looseness=-1 The correlation \eqref{1.1} has been measured in the turbulent liquid crystal experiment by Ta\-keu\-chi and Sano~\cite{TS12} and is also determined numerically by Singha~\cite{Sin05} (for the step case) and Ta\-keu\-chi~\cite{Tak12} in the closely related Eden cluster growth. The large $t$ scaling behavior is reported as
\begin{gather}\label{1.2}
C^\mathrm{f\/lat}(t_0,t) \simeq (t_0)^{4/3} t^{-2/3},\qquad C^\mathrm{step}(t_0,t) \simeq (t_0)^{2/3},
\end{gather}
where we ignored the model-dependent prefactors, see \cite[Section~2.6]{TS12} for more details. Thus in the curved case the correlation of the unscaled height function does not decay to $0$ for large $t$, which is surprising at f\/irst sight. The rough explanation is as follows (see also~\cite{KK99}): In the f\/lat case the height $h(0,t)$ depends on the nucleation events in the backward light cone with base points~$x$ such that $|x|\leq t^{2/3}$ and so does $h(0,t_0)$ with $ |x| \leq (t_0)^{2/3}$. On the other side, in the curved case the domain of dependence has the form of a cigar of width~$t^{2/3}$, resp.~$(t_0)^{2/3}$, since at short times only the few nucleation events close to the initial seed are available. Estimating the overlap in each case results in the distinct behavior as stated in~\eqref{1.2}.

In our contribution we consider the covariance
\begin{gather*}
C_t^\diamond(\tau) = t^{-2/3}C^\diamond(\tau t,t)
\end{gather*}
 rescaled according to the KPZ scaling theory. Thus one expects the limit
 \begin{gather*}
\lim_{t \to \infty} C_t^\diamond(\tau) = C^\diamond(\tau)
\end{gather*}
to exist. Without loss of generality one may set $0 \leq \tau \leq 1$. To study $C^\diamond(\tau)$, we consider last passage percolation (LPP) as a particular model in the KPZ universality class. In this model at zero temperature, the height function is represented through the energy of an optimal directed polymer in a random medium, which is tightly related with the totally asymmetric simple exclusion process (TASEP), see Section~\ref{sec2}. We f\/irst obtain an expression for $C^\diamond(\tau)$ based on a variational problem involving two independent Airy processes. This looks complicated, but we succeed in studying the power law behavior of $C^\diamond(\tau)$ for $\tau$ close to $0$ and $1$, see~(\ref{eq3.4}),~(\ref{eq3.5}). In the f\/irst limit our result is in agreement with the behavior stated in~\eqref{1.2}. For stationary initial conditions we even obtain the entire limiting $C^\mathrm{stat}(\tau)$. Proving our result mathematically rigorously is technically dif\/f\/icult and goes beyond the scope of this paper.

An alternative approach comes from switching to local slopes, $\partial_x h(x,t)$, which are then governed by a type of stochastic particle dynamics. For example, the slope of the single step model is equivalent to the TASEP. The process $t \mapsto \partial_t h(0,t)=\mathsf{j}(t)$ is stationary and the covariance $\Cov(\mathsf{j}(t),\mathsf{j}(t'))$ depends only on~$t-t'$. In the particle picture $\Cov(\mathsf{j}(t),\mathsf{j}(0))$ is the correlation of the current (density) across the origin. We argue that $\int_\R dt \Cov(\mathsf{j}(t),\mathsf{j}(0))=0$ and $\Cov(\mathsf{j}(t),\mathsf{j}(0)) \simeq -|t|^{-4/3}$ for large $|t|$, see (\ref{3.23}). Thereby we arrive at an expression for~$C^\mathrm{stat}(\tau)$ which is identical to the one obtained by the LPP method. In fact, $C^\mathrm{stat}(\tau)$ equals the covariance of fractional Brownian motion with Hurst exponent~$\tfrac{1}{3}$. However, since the rescaled height function is expected to converge to a limit with Baik--Rains distribution, the limiting height process cannot be Gaussian (this is proven for a few models~\cite{BCFV14, FS05a,FSW15,PS02b}).

Our contribution consists of three parts. In Section~\ref{sec2} we investigate $C^\diamond(\tau)$ in the framework of directed polymers. In Section~\ref{sec3} we study the current time correlations for stationary lattice gases and in Section~\ref{sec4} we report on Monte-Carlo simulations of the TASEP in support of our theoretical f\/indings.

\section[Variational formulas for the universal part of the two-time distribution]{Variational formulas for the universal part\\ of the two-time distribution}\label{sec2}

As a model in the KPZ universality class we consider the totally asymmetric simple exclusion process (TASEP). Particle conf\/igurations are denoted by \mbox{$\eta \in \{0,1\}^\mathbb{Z}$}, where $\eta_j = 1$ stands for a~particle at lattice site $j$ and $\eta_j = 0$ for site $j$ being void. Particles jump independently one step to the right after an exponentially distributed waiting time and subject to the exclusion rule. Equivalently the exchange rate between sites $j$ and $j+1$ takes the form $c_{j,j+1}(\eta) = \eta_j(1 - \eta_{j+1})$. The particle conf\/iguration at time $t$ is denoted by $\eta(t)$. Of central interest is the height function,~$h(j,t)$, def\/ined through\footnote{In the literature the height function is mostly def\/ined to be twice the one def\/ined in this paper. As we will discuss also the particle current, in our context it seems to be more natural to avoid unnecessary factors of $2$ relating the two quantities.}
\begin{gather}\label{eq3.0}
h(j,t) = \begin{cases}
\displaystyle J(t) + \sum_{i = 1}^j \tfrac{1}{2}(1 - 2 \eta_i(t)), & \mbox{if $j \geq 1$},\\
J(t), & \mbox{if $j =0$},\\
\displaystyle J(t) - \sum_{i = j+1}^0 \tfrac{1}{2}(1 - 2 \eta_i(t)), & \mbox{if $j \leq -1$},
\end{cases}
\end{gather}
where $J(t)$ is the particle current across the bond $(0,1)$ integrated over the time interval $[0,t]$. Note that $h(0,0)= 0$. We study the TASEP because it allows for a simple mapping to last passage percolation (LPP), which will be the main technical tool in this section.

We will study the three dif\/ferent initial conditions mentioned in the introduction:
\begin{itemize}\itemsep=0pt
 \item[(i)] step initial conditions, $\eta=\Id_{\Z_-}$,
 \item[(ii)] f\/lat initial conditions with density $\tfrac{1}{2}$, $\eta=\Id_{2\Z}$,
 \item[(iii)] stationary initial conditions with density $\tfrac{1}{2}$, i.e., $\eta$ is distributed according to $ \nu_{1/2}$, where $\nu_{\rho}$ is the Bernoulli product measure with density $\rho$.
\end{itemize}
Density $\tfrac{1}{2}$ is chosen for convenience, since in this case the characteristic line has velocity $0$.

For these three initial conditions we would like to understand the scaling limit
\begin{gather}\label{eq3.1}
\tau\mapsto \RP^\diamond(\tau) = \lim_{t\to\infty} -2^{4/3} t^{-1/3} \big(h(0,\tau t)-\tfrac{1}{4}\tau t\big),
\end{gather}
which def\/ines $\RP^\diamond(\tau)$, $\tau \geq 0$, as a stochastic process in $\tau$ (provided the limit exists). $\tau$~is a fraction of the physical time $t$ and the asymptotic mean has been subtracted. The fact that the sca\-ling~(\ref{eq3.1}) should give a non-trivial limit process is due to the slow-decorrelation phenomenon, namely that along special space-time paths, f\/luctuations of order $t^{1/3}$ occurs only over a macroscopic time scale. The special paths are the characteristics of the PDE describing the macroscopic evolution of the particle density~\cite{CFP10b, Fer08}.

Up to model dependent scale factors, the limit processes are expected to be universal, meaning that the limit is the same for any model
in the KPZ universality class. In case the particular initial condition has to be specif\/ied, a superscript is added as $\RP^{\rm step}$, $\RP^{\rm f\/lat}$, $\RP^{\rm stat}$, respectively. The one-point distribution of these processes is well-known~\cite{BR00,BR99b,Jo00b,PS02b} and given by
\begin{gather*}
\Pb\big(\RP^{\rm step}(1)\leq s\big)=F_{\rm GUE}(s),\\
\Pb\big(\RP^{\rm f\/lat}(1)\leq s\big)=F_{\rm GOE}\big(2^{2/3}s\big),\\
\Pb\big(\RP^{\rm stat}(1)\leq s\big)=F_{\rm BR}(s),
\end{gather*}
see Appendix~\ref{app} for their def\/inition. We denote by $\xi_{\rm GUE}$, $\xi_{\rm GOE}$, and $\xi_{\rm BR}$ random variables distributed according to GOE/GUE Tracy--Widom distribution and the Baik--Rains distribution respectively.

For the spatial argument, the corresponding scaling limit reads
\begin{gather}\label{eq3.1a}
w\mapsto \mathcal{Y}^\diamond(w) =\lim_{t\to\infty} -2^{4/3}t^{-1/3} \big(h\big(w2^{1/3}t^{2/3}, t\big)-\tfrac{1}{4}t\big)
\end{gather}
with $w\in \mathbb{R}$. For f\/lat and stationary initial conditions, convergence has been proved in the sense of f\/inite-dimensional distribution~\cite{BFP09,BFPS06, Sas05}. For step initial condition weak*-convergence has been proved in~\cite{Jo03b}. More specif\/ically, one has \mbox{$\mathcal{Y}^{\rm step}(w) = \mathcal{A}_2(w)-w^2$}, \mbox{$\mathcal{Y}^{\rm f\/lat}(w) = 2^{1/3} \mathcal{A}_1(2^{-2/3}w)$}, and \mbox{$\mathcal{Y}^{\rm stat}(w) = \mathcal{A}_{\rm stat}(w)$}, see also the review~\cite{Fer07}. Again we refer to Appendix~\ref{app} for the def\/inition of these Airy processes.

In Section~\ref{sec2.1} we will argue that the joint distribution of $\RP^\diamond(\tau)$ and $\RP^\diamond(1)$ can be expressed through a suitable variational formula, involving two independent copies of $\mathcal{Y}^\circ(w)$, with $\circ\in\{{\rm step},{\rm f\/lat},{\rm stat}\}$ depending on the cases. Unfortunately, it is not so straightforward to extract some useful information from these formulas. Hence we f\/irst try to
study the covariance
\begin{gather*}
C^\diamond(\tau):=\Cov\big(\RP^\diamond(\tau),\RP^\diamond(1)\big)=\E\big(\RP^\diamond(\tau)
\RP^\diamond(1)\big)-\E\big(\RP^\diamond(\tau)\big)\E\big(\RP^\diamond(1)\big) .
\end{gather*}
The parameter $\tau$ can be restricted to the interval $[0,1]$, since the case $\tau >1$ is recovered by a~trivial scaling
from the fact that $\RP^\diamond(\tau)$ is given through the limit~(\ref{eq3.1}). As will be seen from the explicit formula for the stationary case or from the numerical simulation in the other cases, for~$\tau$ away from $0$, $1$, $C^\diamond(\tau)$ looks smooth and strictly increasing, but shows interesting scaling behavior close to the boundary points of this interval.
As one of our main results we determine the respective scaling exponents. For $\tau\to 0$ we obtain
\begin{gather}\label{eq3.4}
C^{\rm step}(\tau) = \OR\big(\tau^{2/3}\big) ,\qquad
C^{\rm f\/lat}(\tau) = \OR\big(\tau^{4/3}\big) ,
\end{gather}
and for $\tau\to 1$ we obtain\footnote{The coef\/f\/icient in front of $(1-\tau)^{2/3}$ for the f\/lat case was conjectured by Takeuchi in~\cite{Tak13} and verif\/ied experimentally in his context.}
\begin{gather}
C^{\rm step}(\tau) = \Var(\xi_{\rm GUE})-\tfrac12\Var(\xi_{\rm BR}) (1-\tau)^{2/3}+\Or(1-\tau),\nonumber\\
C^{\rm f\/lat}(\tau) = 2^{-4/3}\Var(\xi_{\rm GOE})-\tfrac12\Var(\xi_{\rm BR}) (1-\tau)^{2/3}+\Or(1-\tau).\label{eq3.5}
\end{gather}
This implies that for the normalized correlation function $A^\diamond(\tau):=C^\diamond(\tau)/C^\diamond(1)$ we have
\begin{gather*}
A^\diamond(\tau) = 1-c^\diamond (1-\tau)^{2/3}+\Or(1-\tau)
\end{gather*}
as $\tau\to 1$, where
\begin{gather*}
c^{\rm step}=\frac{\Var(\xi_{\rm BR})}{2\Var(\xi_{\rm GUE})}\simeq 0.707,\qquad
c^{\rm f\/lat}= \frac{\Var(\xi_{\rm BR})}{2^{-1/3}\Var(\xi_{\rm GOE})}\simeq 0.901.
\end{gather*}
For the stationary case, we obtain the exact expression
\begin{gather}\label{eq2.10}
C^{\rm stat}(\tau)=\Var(\xi_{\rm BR}) \tfrac{1}{2} \big(1+\tau^{2/3}-(1-\tau)^{2/3}\big).
\end{gather}

The behavior close to $\tau=1$ is based on the same reasoning in all three cases. As key ingredient we use that the limit processes $\mathcal{Y}^\diamond$ def\/ined in (\ref{eq3.1a}) are locally Brownian~\cite{CH11,FSW15, Ha07,PS02,QR12}. Close to $\tau = 0$, step and stationary initial conditions exhibit the same scaling exponent. Interestingly, the $\OR(\tau^{2/3})$ behavior relies on two very distinct mechanisms: for the step it is due to the correlations generated at small times, while for the stationary case it is due to the randomness of the initial conditions.

\subsection{TASEP and LPP}\label{sec2.1}
Let us f\/irst recall the relation between TASEP and LPP. A last passage percolation (LPP) model on $\Z^2$ with independent random variables $\{\omega_{i,j},\,i,j\in\Z\}$ is the following. An \emph{up-right path} $\pi=(\pi(0),\pi(1),\ldots,\pi(n))$ on $\Z^2$ from a point~$A$ to a point~$E$ is a sequence of points in $\Z^2$ with \mbox{$\pi(k+1)-\pi(k)\in \{(0,1),(1,0)\}$}, with $\pi(0)=A$ and $\pi(n)=E$, and where $n$ is called the length $\ell(\pi)$ of $\pi$. Now, given a set of points $S_A$, one def\/ines the last passage time $L_{S_A\to E}$ as
\begin{gather}\label{eq3.2}
L_{S_A\to E}=\max_{\begin{subarray}{c}\pi\colon A\to E\\A\in S_A\end{subarray}} \sum_{1\leq k\leq \ell(\pi)} \omega_{\pi(k)}.
\end{gather}
Finally, we denote by $\pi^{\max}_{S_A\to E}$ any maximizer of the last passage time $L_{S_A\to E}$. For continuous random variables, the maximizer is a.s.\ unique.

\looseness=-1 For the TASEP the ordering of particles is preserved. If initially one orders from right to left as
\begin{gather*} \cdots < x_2(0) < x_1(0) < 0 \leq x_0(0)< x_{-1}(0)< \cdots,\end{gather*}
then for all times $t\geq 0$ also $x_{n+1}(t)<x_n(t)$, $n\in\Z$. The $\omega_{i,j}$ in the LPP is the waiting time of particle $j$ to jump from site $i-j-1$ to site $i-j$.
By def\/inition $\omega_{i,j}$ are ${\rm exp}(1)$ i.i.d.\ random variables. Let $S_A=\{(u,k)\in\Z^2\colon u=k+x_k(0),\, k\in\Z\}$. Then
\begin{gather*}
\Pb(L_{{S_A}\to (m,n)}\leq t)=\Pb(x_n(t)\geq m-n).
\end{gather*}
Further, for $m=n$,
\begin{gather*}
\Pb(L_{{S_A}\to (n,n)}\leq t)=\Pb(x_n(t)\geq 0) = \Pb(J(t)\geq n).
\end{gather*}
In particular, for the initial conditions under consideration, the set $S_A$ is given by
\begin{itemize}\itemsep=0pt
\item[(i)] Step initial conditions: $S_A=\{(0,0)\}$.
\item[(ii)] Flat initial conditions with density $\tfrac{1}{2}$: $S_A={\cal L}=\{(i,j)\,|\, i+j=0\}$.
\item[(iii)] Stationary \looseness=-1 initial conditions with density $\tfrac{1}{2}$: $S_A=\tilde{\cal L}$ is a two-sided simple symmetric random walk passing through the origin and rotated by~$\pi/4$. Using Burke's property~\cite{Bur56} one can equivalently replace all the randomness which is above the random line $\tilde{\cal L}$ but outside the f\/irst quadrant by exponentially distributed random variables with para\-me\-ter~$\tfrac{1}{2}$ only along the bordering lines $\{(i,-1),\, i\geq 0\}$ and $\{(-1,i,),\, i\geq 0\}$, see~\cite{PS02b} for more details.
\end{itemize}
See Fig.~\ref{FigLPP} for an illustration.
\begin{figure}[t]\centering
\includegraphics{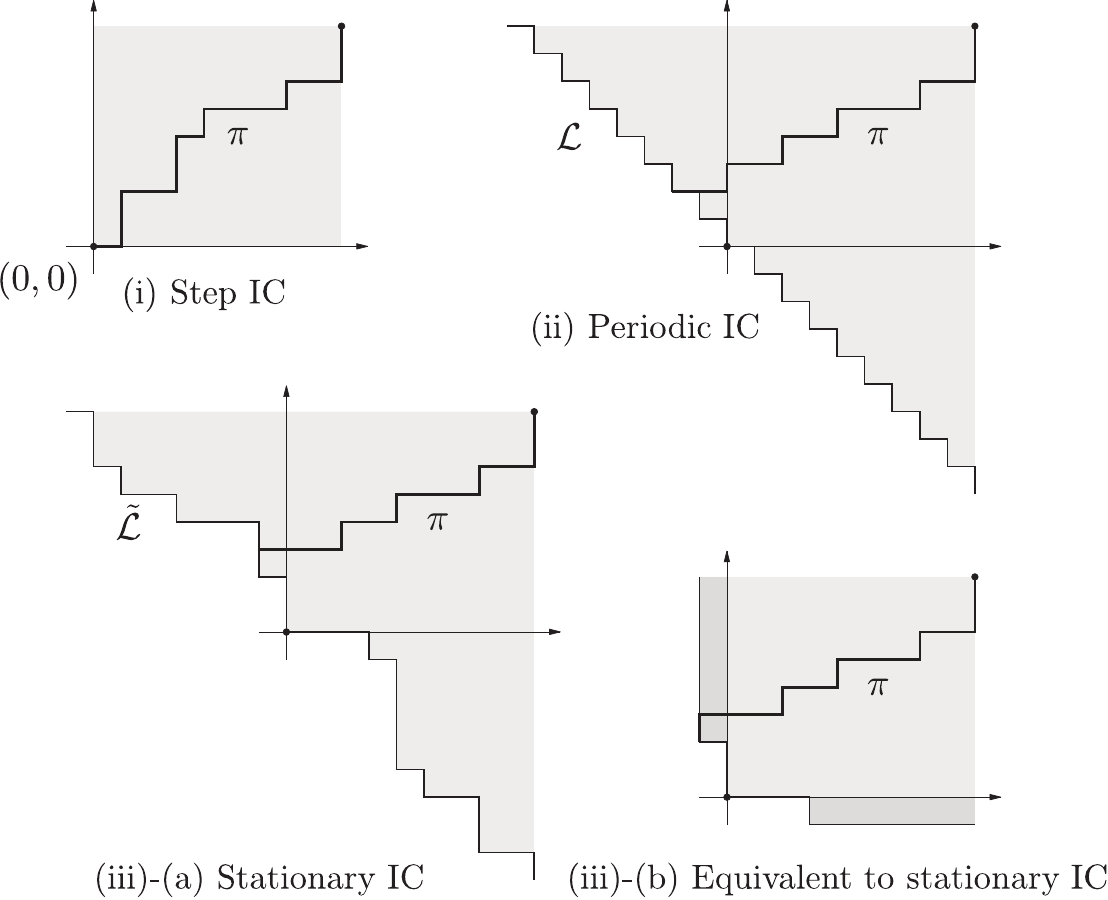}
\caption{Last passage percolation settings corresponding to TASEP with (i) step, (ii) periodic and (iii)~stationary initial conditions. The random variables in the gray regions are $\exp(1)$ i.i.d., while in the dark gray they are $\exp(2)$ i.i.d.
In (iii)-(b) the blank regions at the boundary have a length which is i.i.d. geometric of mean~1.}\label{FigLPP}
\end{figure}

\subsection{Step initial conditions}
TASEP with step initial conditions corresponds to the point-to-point problem in the LPP picture, see Fig.~\ref{FigLPP}(i). In this framework, consider $A_\tau=(\tau t/4,\tau t/4)$ and $I_\tau(u)=A_\tau+u (\tau t/2)^{2/3} (1,-1)$. Then as $t\to\infty$ one has~\cite{BF07,CFP10a, Jo03b}
\begin{gather*}
\frac{L_{0\to A_\tau}-\tau t}{2^{2/3} t^{1/3}}\simeq \tau^{1/3} {\cal A}_2(0),\\
\frac{L_{0\to I_\tau(u)}-\tau t}{2^{2/3} t^{1/3}} \simeq \tau^{1/3} \big({\cal A}_2(u)-u^2\big),\\
\frac{L_{I_\tau(u)\to A_1}-(1-\tau) t}{2^{2/3} t^{1/3}} \simeq (1-\tau)^{1/3} \big[\tilde{\cal A}_2\big(u \hat{\tau}^{2/3}\big) -
\big(u \hat{\tau}^{2/3}\big)^2\big],
\end{gather*}
where ${\cal A}_2$ and $\tilde {\cal A}_2$ are two independent Airy$_2$ processes. These identities are understood for f\/ixed~$\tau$, where the f\/irst is convergence of random variables, while the last two identities hold as processes in~$u$. Also we introduced the convenient shorthand~$\hat{\tau}
= \tau/(1-\tau)$. Using~(\ref{eq3.1}) and~(\ref{eq3.2}) we thus conclude
\begin{gather*}
\RP^{\rm step}(\tau)=\lim_{t\to\infty} \frac{L_{0\to A_\tau}-\tau t}{2^{2/3} t^{1/3}}.
\end{gather*}
Therefore
\begin{gather*}
\RP^{\rm step}(\tau)=\tau^{1/3} {\cal A}_2(0)
\end{gather*}
and, using the relation $L_{0\to A_1}=\max_{u} (L_{0\to I_\tau(u)}+L_{I_\tau(u)\to A_1})$, also
\begin{gather}\label{eqRP1}
\RP^{\rm step}(1)=\tau^{1/3}\max_{u\in\R}\big\{{\cal A}_2(u)-u^2+\hat\tau^{-1/3} \tilde {\cal A}_2\big(u\hat{\tau}^{2/3}\big)-u^2 \hat{\tau}\big\}.
\end{gather}
Together these formulas are a tool for determining the joint distribution of $\RP^{\rm step}(\tau)$, $\RP^{\rm step}(1)$.

\textbf{Limit} $\tau\to 0$. First of all, as $\tau\to 0$, as a process in $u$,
\begin{gather}\label{eqA2BM}
\hat{\tau}^{-1/3} \big(\tilde{\cal A}_2\big(u \hat{\tau}^{2/3}\big)-\tilde{\cal A}_2(0)\big)\simeq \sqrt{2} B(u),
\end{gather}
where $B$ is a standard Brownian motion~\cite{CH11,Ha07,PS02} (with standard meaning with normalization $\Var(B(u))=u$). Further, for the two terms proportional to~$u^2$, the right term is of order~$\tau$ smaller than the left one. Therefore the maximum in~(\ref{eqRP1}) is taken at $u=\OR(1)$ and consequently as $\tau\to 0$ we have
\begin{gather*}
\begin{split}
& C^{\rm step}(\tau)=\Cov\big(\RP^{\rm step}(\tau),\RP^{\rm step}(1)\big)\\
& \hphantom{C^{\rm step}(\tau)}{} \simeq \tau^{2/3}\Cov\big({\cal A}_2(0),\max_{u\in\R}\big\{{\cal A}_2(u)-u^2+\sqrt{2}B(u)\big\}+\hat\tau^{-1/3}\tilde{\cal A}_2(0)\big),
\end{split}
\end{gather*}
where the processes ${\cal A}_2$ and $B$ are independent, and $B$ is independent of $\tilde{\cal A}_2(0)$. Since ${\cal A}_2(0)$ and $\tilde{\cal A}_2(0)$ are independent, their covariance is zero.

To understand what happens, we rewrite the expectation in the covariance as the expectation of the conditional expectation with respect to the Brownian motion $B$, namely
\begin{gather*}
 \Cov\big({\cal A}_2(0),\max_{u\in\R}\big\{{\cal A}_2(u)-u^2+\sqrt{2}B(u)\big\}\big)\\
\qquad {} = \E\Big[\Cov\Big({\cal A}_2(0),\max_{u\in\R}\big\{{\cal A}_2(u)-u^2+\sqrt{2}B(u)\big\}\big| B \Big)\Big].
\end{gather*}
For typical realizations of $B$, the maximum is reached for $u$ of order $1$ (for $B=0$ there is an explicit formula, see~\cite{BKS12,MFQR11,Sch12}). On the other hand, the random variables $\max\limits_{u\in\R}({\cdots})$ and~${\cal A}_2(0)$ are non-trivially correlated. Therefore we conclude $C^{\rm step}(\tau)=\OR(\tau^{2/3})$ as $\tau\to 0$.

\begin{rem}
In the LPP picture, the fact that the maximum is obtained for $u$ of order $1$ is a~consequence of the constraint that the polymer maximizing $L_{0\to A_1}$ starts at the origin.
\end{rem}
\begin{rem}
We have
\begin{gather}
\Cov\Big({\cal A}_2(0),\max_{u\in\R}\big\{{\cal A}_2(u)-u^2+\sqrt{2}B(u)\big\}\Big)\nonumber\\
\qquad{} =\E\Big({\cal A}_2(0)\max_{u\in\R}\big\{{\cal A}_2(u)-u^2+\sqrt{2}B(u)\big\}\Big),\label{eq2.21}
\end{gather}
where we used the fact that $\E({\cal A}_{\rm stat}(0))=0$ and the identity~\cite{QR13}
\begin{gather}\label{eq2.23}
\RP^{\rm stat}(1)={\cal A}_{\rm stat}(0)\stackrel{d}{=}\max_{v\in\R}\big\{{\cal A}_2(v)-v^2+\sqrt{2} B(v)\big\}
\end{gather}
in distribution, where the Airy$_2$ process ${\cal A}_2$ and the Brownian motion $B$ are independent. The joint distribution of the two random variables in~(\ref{eq2.21}) might be obtained analytically from the formulas in~\cite{Joh15} and~\cite{Dot13}.
\end{rem}

\textbf{Limit} $\tau\to 1$. In this case, the maximum in (\ref{eqRP1}) is achieved for $u=\OR((1-\tau)^{2/3})$ as can one see for instance by symmetry of the point-to-point problem. Therefore let us set $v=u \hat{\tau}^{2/3}$ so that now
\begin{gather}\label{eq2.20}
\RP^{\rm step}(1)
=(1-\tau)^{1/3}\max_{v\in\R}\big\{ \hat\tau^{1/3} {\cal A}_2\big(v\hat{\tau}^{-2/3}\big)-v^2\hat{\tau}^{-1}+\big( \tilde {\cal A}_2(v)-v^2\big)\big\}.
\end{gather}
To argue about the behavior for $\tau\to 1$, we will use the convergence of the Airy$_2$ process to Brownian motion (see (\ref{eqA2BM})) and we use the identity
\begin{gather*}
C^{\rm step}(\tau)=\tfrac12 \Var(\RP^{\rm step}(1))+\tfrac12 \Var(\RP^{\rm step}(\tau))-\tfrac12 \E\big((\RP^{\rm step}(\tau)-\RP^{\rm step}(1))^2\big)\\
\hphantom{C^{\rm step}(\tau)}{} =\tfrac12(1+\tau^{2/3})\Var(\RP^{\rm step}(1))-\tfrac12 \E\big((\RP^{\rm step}(\tau)-\RP^{\rm step}(1))^2\big).
\end{gather*}
Now, by (\ref{eq2.20}) and $\RP^{\rm step}(\tau)=(1-\tau)^{1/3} \hat\tau^{1/3} {\cal A}_2(0)$, we have
\begin{gather*}
\RP^{\rm step}(1)-\RP^{\rm step}(\tau) =(1-\tau)^{1/3}\max_{v\in\R}\big\{\hat\tau^{1/3} \big[{\cal A}_2\big(v \hat\tau^{-2/3}\big)\!-{\cal A}_2(0)\big]+\tilde {\cal A}_2(v)-v^2 \big(1+\hat\tau^{-1}\big)\big\},
\end{gather*}
where ${\cal A}_2$ and $\tilde {\cal A}_2$ are independent Airy$_2$ processes. In the $\tau\to 1$ limit, using (\ref{eqA2BM}) the f\/irst term becomes $\sqrt{2}B(v)$ and since the maximum is obtained for $v$ of order one, the term $v^2\hat\tau^{-1}$ should be at most a correction of order $\Or(1-\tau)$. (\ref{eq2.23}) gives us
\begin{gather*}
C^{\rm step}(\tau)\simeq \tfrac12\big(1+\tau^{2/3}\big)\Var(\RP^{\rm step}(1))-\tfrac12 (1-\tau)^{2/3}\Var\big(\RP^{\rm stat}(1)\big)+ \Or(1-\tau),
\end{gather*}
where we used the property that ${\cal A}_{\rm stat}(0)$ has mean zero.

\begin{rem}
To make the present result into a theorem one has to control the convergence of the Airy process to Brownian motion. In recent work in progress, Corwin and Hammond establish rigorously the behavior close to $\tau=0$ and $\tau=1$ for the point-to-point problem~\cite{CH16}.
\end{rem}

\subsection{Flat initial conditions}
TASEP with f\/lat initial conditions corresponds to the point-to-line problem in the LPP picture, as illustrated in Fig.~\ref{FigLPP}(ii). Consider $A_\tau=(\tau t/4,\tau t/4)$ and \mbox{$I_\tau(u)=A_\tau+u (\tau t/2)^{2/3} (1,-1)$}. From~\cite{BFPS06,CFP10a}, we know that by setting $c=2^{1/3}$, in the $t\to\infty$ limit we have
\begin{gather*}
\frac{L_{{\cal L}\to A_\tau}-\tau t}{2^{2/3}t^{1/3}} \simeq c \tau^{1/3} {\cal A}_1(0),\\
\frac{L_{{\cal L}\to I_\tau(u)}-\tau t}{2^{2/3}t^{1/3}} \simeq c \tau^{1/3} {\cal A}_1\big(c^{-2} u\big),\\
\frac{L_{I_\tau(u)\to A_1}-(1-\tau) t}{2^{2/3}t^{1/3}} \simeq (1-\tau)^{1/3} \big[\tilde {\cal A}_2\big(u \hat{\tau}^{2/3}\big)-\big(u \hat{\tau}^{2/3}\big)^2\big],
\end{gather*}
where the Airy$_1$ process ${\cal A}_1$ is independent of the Airy$_2$ process $\tilde {\cal A}_2$. As before,
the f\/irst identity is understood for f\/ixed $\tau$, while the last two identities hold as processes in $u$.
We have
\begin{gather*}
\RP^{\rm f\/lat}(\tau)=\lim_{t\to\infty} \frac{L_{{\cal L}\to A_\tau}-\tau t}{2^{2/3} t^{1/3}}
\end{gather*}
and thus
\begin{gather*}
\RP^{\rm f\/lat}(\tau)=c\tau^{1/3} {\cal A}_1(0).
\end{gather*}
Further, using the relation $L_{{\cal L}\to A_1}=\max_{u} (L_{{\cal L}\to I_\tau(u)}+L_{I_\tau(u)\to A_1})$, we obtain
\begin{gather*}
\RP^{\rm f\/lat}(1)=\tau^{1/3}\max_{u\in\R}\big\{ c {\cal A}_1\big(c^{-2} u\big)+\hat\tau^{-1/3} \tilde {\cal A}_2\big( u \hat\tau^{2/3}\big)-u^2\hat\tau\big\}.
\end{gather*}

\textbf{Limit} $\tau\to 0$. Unlike for step initial conditions, this time the quadratic term responsible for the localization of the maximizer over a distance of order $1$ (in the $u$ variable) is absent. This implies that the maximization no longer occurs for $u$ of order $1$. Rather, from~\cite{Jo00,MFQR11} we know that the point-to-line maximizer starts from the line ${\cal L}$ at a distance of order $t^{2/3}$ from the origin. As a~consequence the maximization will occur typically at values $u=\OR(\tau^{-2/3})$. Therefore
\begin{gather*}
C^{\rm f\/lat}(\tau) =\tau^{2/3}\Cov\Big({\cal A}_1(0),\max_{u\in\R}\big\{ c {\cal A}_1\big(c^{-2} u\big)+ \hat\tau^{-1/3} \tilde {\cal A}_2\big( u \hat{\tau}^{2/3}\big)-u^2 \hat{\tau} \big\}\Big) \\
\hphantom{C^{\rm f\/lat}(\tau)}{} = \tau^{2/3} \E\Big[\Cov\Big({\cal A}_1(0),\max_{u\in\R}\big\{ c {\cal A}_1\big(c^{-2} u\big)+\hat\tau^{-1/3} \tilde {\cal A}_2\big( u \hat{\tau}^{2/3}\big)-u^2 \hat{\tau}\big\}\Big) \big| \tilde {\cal A}_2\Big].
\end{gather*}
To understand the behavior at small values of $\tau$ of the covariance between $\RP^{\rm f\/lat}(\tau)$ and $\RP^{\rm f\/lat}(1)$, we need to consider the following two cases (see Fig.~\ref{FigLPPStepFlat} for an illustration).
\begin{figure}[t]\centering
\includegraphics{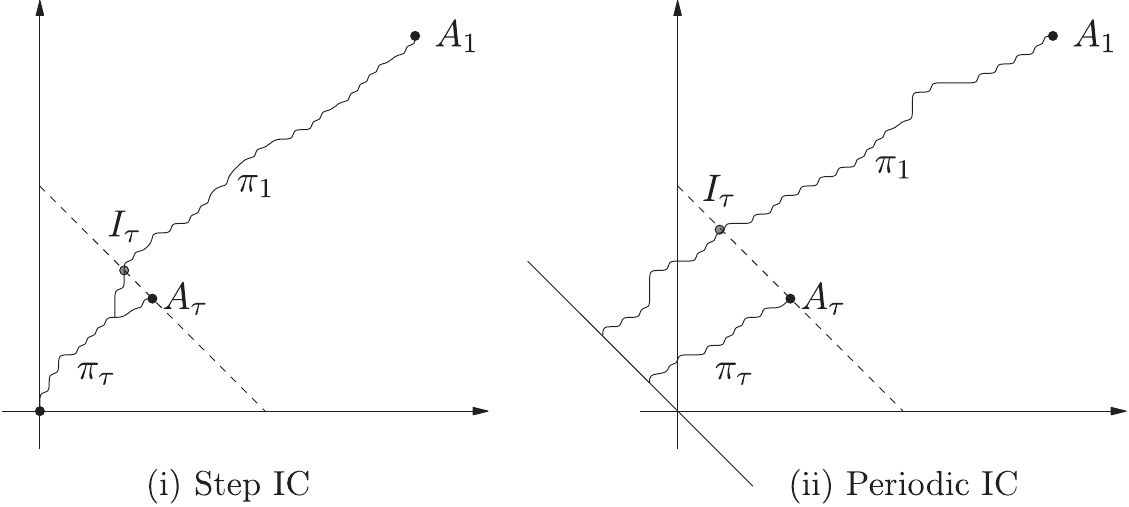}
\caption{The maximizer of the LPP for $A_\tau$ is denoted by $\pi_\tau$, and for $A_1$ by $\pi_1$. The LPP for $A_1$ can be decomposed in the LPP to the dashed line and the one from the dashed line to $A_1$. For periodic initial condition, the probability that $\pi_\tau$ and $\pi_1$ merges is expected to be of order $\OR(\tau^{2/3})$.}\label{FigLPPStepFlat}
\end{figure}

(1) Realizations of $\tilde {\cal A}_2$ such that the maximization occurs for $u\gg 1$. In this case, since the covariance of the Airy$_1$ process ${\cal A}_1$ decays super-exponentially~\cite{BFP08}, the covariance conditioned on those events goes to zero faster than any power of $\tau$.

(2) Realizations of $\tilde {\cal A}_2$ such that the maximization occurs for $u=\OR(1)$. In this case, the covariance conditioned on those events is of order $\OR(\tau^{2/3})$ by the same argument as for step initial conditions. The only minor dif\/ference is to replace ${\cal A}_2(u)-u^2$ by $c {\cal A}_1(c^{-2} u)$.

The f\/irst situation occurs with probability of order $1-\OR(\tau^{2/3})$, while the second case only with probability $\OR(\tau^{2/3})$. This is due to the superdif\/fusive transversal f\/luctuations of the maximizers (compare with the point-to-point transversal f\/luctuations in Poisson points see~\cite{Jo00} and \cite[Section~9]{BSS14} for a ref\/ined result).
Therefore as $\tau\to 0$,
\begin{gather*}
C^{\rm f\/lat}(\tau)=\Cov(\RP(\tau),\RP(1))=\OR\big(\tau^{4/3}\big).
\end{gather*}

\textbf{Limit} $\tau\to 1$. We use the same argument as for the step-initial condition. (\ref{eq2.20}) is replaced by
\begin{gather*}
\RP^{\rm f\/lat}(1)
=(1-\tau)^{1/3}\max_{v\in\R}\big\{ \hat\tau^{1/3} {\cal A}_1\big(v\hat{\tau}^{-2/3}\big)+\big( \tilde {\cal A}_2(v)-v^2\big)\big\}.
\end{gather*}
Thus we get
\begin{gather*}
C^{\rm f\/lat}(\tau)=\tfrac12(1+\tau^{2/3})\Var\big(\RP^{\rm f\/lat}(1)\big)-\tfrac12 \E\big(\big(\RP^{\rm f\/lat}(\tau)-\RP^{\rm f\/lat}(1)\big)^2\big).
\end{gather*}
Now,
\begin{gather*}
\RP^{\rm f\/lat}(\tau)-\RP^{\rm step}(1) =(1-\tau)^{1/3}\max_{v\in\R}\big\{\hat\tau^{1/3} \big[{\cal A}_1\big(v \hat\tau^{-2/3}\big)-{\cal A}_1(0)\big]+\tilde {\cal A}_2(v)-v^2\big\},
\end{gather*}
where the two Airy processes, ${\cal A}_1$ and $\tilde {\cal A}_2$, are independent. Using the property that the Airy$_1$ process is locally Brownian~\cite{QR12}, one concludes that
\begin{gather*}
C^{\rm f\/lat}(\tau)\simeq \tfrac12\big(1+\tau^{2/3}\big)\Var\big(\RP^{\rm f\/lat}(1)\big)-\tfrac12 (1-\tau)^{2/3} \Var\big(\RP^{\rm stat}(1)\big)+ \Or(1-\tau).
\end{gather*}

\subsection{Stationary initial conditions}

{\sloppy For the stationary initial conditions we employ the LPP with boundary conditions, see Fig.~\ref{FigLPP}(iii)(b) for an illustration, and denote the corresponding maximal last passage time by~$L^{\cal B}$. Let $A_\tau=(\tau t/4,\tau t/4)$ and $I_\tau(u)=A_\tau+u (\tau t/2)^{2/3} (1,-1)$. Then from~\cite{BFP09, SI04b} we know that in the limit $t\to\infty$ one has
\begin{gather*}
\frac{L^{\cal B}_{0\to A_\tau}-\tau t}{2^{2/3} t^{1/3}} \simeq \tau^{1/3} {\cal A}_{\rm stat}(0),\\
\frac{L^{\cal B}_{0\to I_\tau(u)}-\tau t}{2^{2/3} t^{1/3}} \simeq \tau^{1/3} {\cal A}_{\rm stat}(u),\\
\frac{L_{I_\tau(u)-A_1}-(1-\tau) t}{2^{2/3} t^{1/3}} \simeq (1-\tau)^{1/3} \big[\tilde{\cal A}_2\big(u \hat{\tau}^{2/3}\big) -
\big(u \hat{\tau}^{2/3}\big)^2\big],
\end{gather*}
where the processes ${\cal A}_{\rm stat}$ and $\tilde {\cal A}_2$ are independent. As before,
the f\/irst identity is understood for f\/ixed~$\tau$, while the last two identities hold as processes in~$u$.

}

Further it holds
\begin{gather*}
\RP^{\rm stat}(\tau)=\lim_{t\to\infty} \frac{L^{\cal B}_{0\to A_\tau}-\tau t}{2^{2/3} t^{1/3}},\qquad
\RP^{\rm stat}(\tau)=\tau^{1/3} {\cal A}_{\rm stat}(0),
\end{gather*}
and, using the relation $L^{\cal B}_{0\to A_1}=\max_{u} (L^{\cal B}_{0\to I_\tau(u)}+L_{I_\tau(u)\to A_1} )$, we obtain
\begin{gather}\label{eqRP3}
\RP^{\rm stat}(1)=\tau^{1/3}\max_{u\in\R}\big\{{\cal A}_{\rm stat}(u)+\hat\tau^{-1/3} \tilde {\cal A}_2\big(u\hat{\tau}^{2/3}\big)-u^2\hat{\tau}^{-1}\big\}.
\end{gather}

\textbf{Limit} $\tau\to 0$. In the LPP picture with boundary terms, denote by $C_1$ and $C_\tau$ the sites on the boundary at which the maximizers of $L_{(-1,-1)\to A_1}$ and $L_{(-1,-1)\to A_\tau}$ enter into the po\-si\-ti\-ve quadrant. Similarly to f\/lat initial conditions, the maximizer in (\ref{eqRP3}) is attained for $u$ of order~$\OR(\tau^{-2/3})$.

However, this time the correlations do not decay super-exponentially. We have
\begin{gather*}
C^{\rm stat}(\tau) =\tau^{2/3}\Cov\Big({\cal A}_{\rm stat}(0),\max_{u\in\R}\big\{{\cal A}_{\rm stat}(u)+\hat\tau^{-1/3} \tilde {\cal A}_2\big( u \hat{\tau}^{2/3}\big)-u^2\hat{\tau}\big\}\Big) \\
\hphantom{C^{\rm stat}(\tau)}{} = \tau^{2/3}\E\Big[\Cov\Big({\cal A}_{\rm stat}(0),\max_{u\in\R}\big\{{\cal A}_{\rm stat}(u)+\hat\tau^{-1/3} \tilde {\cal A}_2\big( u \hat{\tau}^{2/3}\big)-u^2\hat{\tau}\big\}\Big) \big| \tilde {\cal A}_2\Big].
\end{gather*}
To understand the behavior for the covariance of $\RP(\tau)$ and $\RP(1)$ at small values of $\tau$, we need to consider the following two cases (see Fig.~\ref{FigLPPStationary} for an illustration).
\begin{figure}[t]\centering
\includegraphics{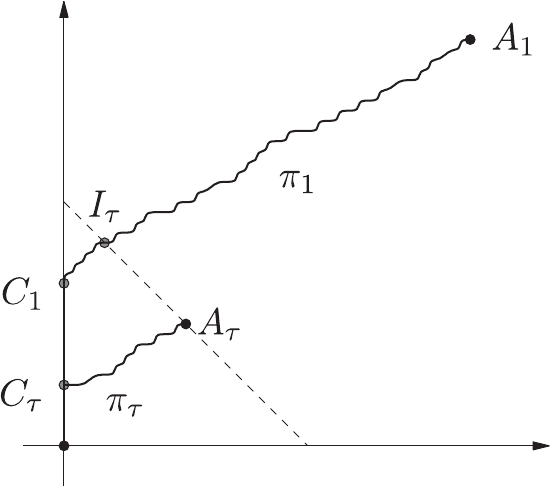}
\caption{The maximizer of the LPP to $A_\tau$, $A_1$ are denoted by $\pi_\tau$, $\pi_1$ respectively. $C_1$ and $C_\tau$ are the points where the maximizers leaves the axis.}\label{FigLPPStationary}
\end{figure}

 (1) Realization of $\tilde {\cal A}_2$ such that the maximization occurs for $u=\OR(1)$. The same argument as for step initial conditions indicates that the covariance conditioned on those events is of order~$\OR(\tau^{2/3})$. Since these events occur with probability of order~$\OR(\tau^{2/3})$, the overall contribution is of order~$\OR(\tau^{4/3})$.

(2) Realizations of $\tilde {\cal A}_2$ such that the maximization occurs for $u\gg 1$. This event occurs with probability $1-\OR(\tau^{2/3})$. The maximizers of $L_{(-1,-1)\to A_1}$ and of $L_{(-1,-1)\to A_\tau}$ use disjoint background noise, except for the randomness on the boundaries (in case they are at the same boundary). Thus in this case the covariance of the LPP to $A_1$ and $A_\tau$ should be as the covariance of the LPP to $C_1$ and $C_\tau$ at leading order.

With this reasoning, one expects that
\begin{gather*}
C^{\rm stat}(\tau) =\Cov\big(\RP^{\rm stat}(\tau),\RP^{\rm stat}(1)\big)\\
\hphantom{C^{\rm stat}(\tau)}{} \simeq \OR(1)\max\big\{\tau^{4/3},t^{-2/3}\Cov(L_{(-1,-1)\to C_\tau},L_{(-1,-1)\to C_1})\big\}.
\end{gather*}
Since the LPP on the boundaries is merely sum of iid random variables, by the central limit theorem, in the $t\to\infty$ limit,
\begin{gather*}
t^{-1/3} L_{(-1,-1)\to (x t^{2/3},-1)} \to 2{\cal B}(x),\qquad
t^{-1/3} L_{(-1,-1)\to (-1,x t^{2/3})} \to 2{\cal B}(-x),
\end{gather*}
where $x\mapsto {\cal B}(x)$ is a two-sided Brownian motion with constant drift. For its covariance,
$\Cov({\cal B}(x),{\cal B}(y))=\tfrac{1+\sgn(xy)}{2} \min\{|x|,|y|\}$ independent of the drift. Finally, since $|C_1|\sim t^{2/3}$ and $|C_\tau|\sim (\tau t)^{2/3}$, we obtain
\begin{gather*}
C^{\rm stat}(\tau)\simeq \OR(1)\max\big\{\tau^{4/3},\tau^{2/3}\big\}=\OR\big(\tau^{2/3}\big).
\end{gather*}

\textbf{Entire $\tau$ interval}. The argument used to determine the $\tau\to 1$ limit in the step and f\/lat initial condition case, can be used to derive a formula for the covariance in the stationary case. (\ref{eq2.20}) is replaced by
\begin{gather*}
\RP^{\rm stat}(1)
=(1-\tau)^{1/3}\max_{v\in\R}\big\{ \hat\tau^{1/3} {\cal A}_{\rm stat}\big(v\hat{\tau}^{-2/3}\big)+\big( \tilde {\cal A}_2(v)-v^2\big)\big\}.
\end{gather*}
Thus we get
\begin{gather*}
C^{\rm stat}(\tau)=\tfrac12\big(1+\tau^{2/3}\big)\Var\big(\RP^{\rm stat}(1)\big)-\tfrac12 \E\big(\big(\RP^{\rm stat}(\tau)-\RP^{\rm stat}(1)\big)^2\big).
\end{gather*}
But now
\begin{gather*}
\RP^{\rm stat}(\tau)-\RP^{\rm stat}(1) =(1-\tau)^{1/3}\max_{v\in\R}\big\{\hat\tau^{1/3} \big[{\cal A}_{\rm stat}\big(v \hat\tau^{-2/3}\big)-{\cal A}_{\rm stat}(0)\big]+\tilde {\cal A}_2(v)-v^2\big\},
\end{gather*}
where the two Airy processes, ${\cal A}_{\rm stat}$ and $\tilde {\cal A}_2$, are independent. For Airy$_{\rm stat}$ the increments are not only locally Brownian, but exactly Brownian. More precisely,
\begin{gather*}
\hat\tau^{1/3} \big[{\cal A}_{\rm stat}\big(v \hat\tau^{-2/3}\big)-{\cal A}_{\rm stat}(0)\big] \stackrel{d}{=} \sqrt{2} B(u),
\end{gather*}
where $B$ is a standard Brownian motion. Then, using the identity (\ref{eq2.23}), we obtain
\begin{gather}\label{eq2.39}
C^{\rm stat}(\tau) = \tfrac12\big(1+\tau^{2/3}-(1-\tau\big)^{2/3})\Var\big(\RP^{\rm stat}(1)\big)
\end{gather}
for $0 \leq \tau \leq 1$.

\section{Current covariance for stationary lattice gases}\label{sec3}
The height function $h(0,t)$ of the TASEP is identical to the time-integrated current across the bond $(0,1)$, denoted by~$J(t)$ in~\eqref{eq3.0}. This suggests to study the covariance of the same observable for a more general class of one-dimensional lattice gases. The mapping to LPP is then lost. On the other hand, in case of stationary initial conditions, one can exploit the local conservation law for the particle number together with space-time stationarity to obtain some information on the current covariance. Thereby we extend the validity of~\eqref{eq2.39}. The covariance of $J(t)$ is identical to the one of fractional Brownian motion in the scaling limit. For reversible models the Hurst parameter is $H = \tfrac{1}{4}$, while for non-reversible lattice gases $H = \tfrac{1}{3}$. In fact, for reversible models it is expected, and proved for particular cases~\cite{DMF02,PS08}, that as a stochastic process~$J(t)$ converges under the appropriate scaling to fractional Brownian motion, which is a Gaussian process. Such a result cannot hold in the non-reversible case, since the large~$t$ distribution of~$J(t)$ is Baik--Rains, as proved for a few models~\cite{BCFV14, FS05a,FSW15,PS02b}.

We consider exclusion processes on $\mathbb{Z}$, for simplicity with nearest neighbor jumps only. They are def\/ined as a generalization
of the TASEP by allowing for an arbitrary exchange rate $c_{j,j+1}(\eta) >0$. For the ASEP the exchange rates are $c_{j,j+1}(\eta) = p\eta_j(1-\eta_{j+1})+q(1- \eta_j)\eta_{j+1}$ with $p+q =1$, $p = \tfrac{1}{2}$ being the reversible SSEP. We assume that $c_{j,j+1}$ has f\/inite range and is invariant under lattice translations. The generator,~$L$, of the corresponding Markov jump process is then def\/ined through
\begin{gather*}
Lf(\eta) = \sum_{j\in\mathbb{Z}}c_{j,j+1}(\eta)\big(f\big(\eta^{j,j+1}\big) - f(\eta)\big)
\end{gather*}
acting on local functions $f$, where $\eta^{j,j+1}$ denotes the conf\/iguration $\eta$ with occupancies at sites~$j$ and $j+1$ exchanged.

We start the dynamics in the steady state. For reversible models, by def\/inition there is a~f\/inite range translation invariant energy function, $H$, such that
\begin{gather*}
c_{j,j+1}(\eta) = c_{j,j+1}\big(\eta^{j,j+1}\big) \mathrm{e}^{-[H(\eta^{j,j+1}) - H(\eta)]}.
\end{gather*}
For given average density, $0 \leq \rho \leq 1$, there is a unique stationary measure, $\mu_\rho$, satisfying $\mu_\rho = \mu_\rho \mathrm{e}^{Lt}$. $\mu_\rho$ is the Gibbs measure for $H - \bar{\mu} \sum_j \eta_j $, where the chemical potential $\bar{\mu}$ has to be adjusted such that the average density equals $\rho$. On the other hand, for non-reversible lattice gases one immediately encounters the long-standing problem to prove the existence of a unique stationary measure at f\/ixed $\rho$. Here we simply assume such a property to be valid, including the exponential space-mixing of~$\mu_\rho$. We use $\E (\cdot)$ as a generic symbol for the process expectation and $\langle \cdot\rangle_\rho$ as expectation with respect to~$\mu_\rho$. For the ASEP the steady state is Bernoulli and obviously our assumptions hold.

Let us consider the empirical current across the bond $(j,j+1)$, denoted by $\mathsf{j}_{j,j+1}(t)$. This is a sequence of $\delta$-functions with weight $1$ for a jump from $j$ to $j+1$ and weight $-1$ for the reverse jump. The time-integrated current across the bond $(j,j+1)$ is then
\begin{gather*}
J_{j,j+1}(t) = \int_0^t ds\, \mathsf{j}_{j,j+1}(s)
\end{gather*}
with the convention $J(t) = J_{0,1}(t)$, $\mathsf{j}(t) = \mathsf{j}_{0,1}(t)$. The average current reads $\mathbb{E} ( \mathsf{j}(t) ) \!= \!\langle c_{0,1}(\eta)(\eta_0 - \eta_1)\rangle_\rho = j(\rho)$.
We also introduce the stationary covariance
\begin{gather*}
S(j,t) = \Cov(\eta_j(t),\eta_0(0)) = \mathbb{E}\big(\eta_j(t)\eta_0(0)\big) - \rho^2.
\end{gather*}
There is a sum rule which connects $S$ with the variance of $J(t)$,
\begin{gather}\label{3.5}
\Var( J(t)) = \sum_{j \in \mathbb{Z}}|j|S(j,t) - \sum_{j \in \mathbb{Z}}|j|S(j,0).
\end{gather}
The proof is deferred to Appendix~\ref{app1}.

Since $J(t)$ has stationary increments, it is convenient to study the correlations of the increments $dJ(t) = \mathsf{j}(t)dt$. As discussed in Appendix~\ref{app1}, the covariance is given by
\begin{gather*}
\Cov\big(\mathsf{j}(t),\mathsf{j}(t')\big) = \langle c_{0,1}\rangle_\rho \delta(t-t')+ h(t-t').
\end{gather*}
For the continuous part we f\/irst def\/ine the generator of time reversed process, $L^\mathrm{R}$, through $\langle f (Lg) \rangle_\rho = \langle (L^\mathrm{R}f)g \rangle_\rho$. Its exchange rates are given by
\begin{gather*}
c^\mathrm{R}_{j,j+1}(\eta) = \frac{\mu_\rho(\eta^{j,j+1})}{\mu_\rho(\eta)} c_{j,j+1}\big(\eta^{j,j+1}\big).
\end{gather*}
Hence the current function across the bond $(j,j+1)$ equals
\begin{gather}r_{j,j+1}(\eta) = c_{j,j+1}(\eta)(\eta_j - \eta_{j+1})
\end{gather}
and the time-reversed current function equals
\begin{gather*}
r^\mathrm{R}_{j,j+1}(\eta) = c^\mathrm{R}_{j,j+1}(\eta^{j,j+1})(\eta_j - \eta_{j+1}).
\end{gather*}
They satisfy $\langle r_{j,j+1}\rangle_\rho = - \langle r^\mathrm{R}_{j,j+1}\rangle_\rho$. Then
\begin{gather}\label{3.20}
h(t) = -\big\langle \big(r^\mathrm{R}_{0,1}- j(\rho)\big)\mathrm{e}^{L|t|}(r_{0,1}- j(\rho))\big\rangle_\rho.
\end{gather}

\subsection{Reversible models}\label{sec3.1}
While our focus is on non-reversible models, it is still instructive to f\/irst explain how fractional Brownian motion appears for reversible lattice gases. Then $r^\mathrm{R}_{j,j+1} = r_{j,j+1}(\eta)$ and the smooth part~$h(t)$ simplif\/ies to
\begin{gather*}
h(t) =- \langle c_{0,1}(\eta) (\eta_0 - \eta_1)\mathrm{e}^{L|t|} c_{0,1}(\eta) (\eta_0 - \eta_1)\rangle_\rho,
\end{gather*}
see Appendix~\ref{app1}. Since $L$ is a symmetric operator in the Hilbert space $L^2(\{0,1\}^\Z,\mu_\rho)$, there exists a spectral measure $\nu$ of f\/inite mass such that
\begin{gather}\label{3.8}
h(t) = - \int_0^\infty \nu(d\lambda) \mathrm{e}^{-\lambda |t|}.
\end{gather}
In particular, $h$ is monotonically increasing with $h(0) = - \langle c_{0,1}(\eta)^2(\eta_0 - \eta_1)^2\rangle_\rho$ and $h(\infty) = 0$.

\looseness=-1
From hydrodynamic f\/luctuation theory~\cite{Cha94,Cha96}, one knows that $S(j,t)$ broadens dif\/fusively as
\begin{gather}\label{3.9}
S(j,t) \simeq \chi (Dt)^{-1/2} f_\mathrm{G}\big((Dt)^{-1/2}j\big)
\end{gather}
with $f_\mathrm{G}$ the standard Gaussian, $D$ a dif\/fusion constant depending on $\rho$, and the susceptibility
\begin{gather*}
\chi = \sum _{j \in \Z}S(j,0).
\end{gather*}
Hence, using \eqref{3.9} for large $t$,
\begin{gather*}
\sum_{j \in \mathbb{Z}}|j|S(j,t) \simeq \chi(Dt)^{1/2}\int_\mathbb{R} dx |x|f_\mathrm{G}(x).
\end{gather*}
Now,
\begin{gather*}
\Var(J(t)) = \langle c_{0,1}\rangle_\rho t + \int_0^tds \int_0^tds' h(s-s').
\end{gather*}
The sum rule \eqref{3.5} implies a variance of order~$\sqrt{t}$. Thus to cancel the leading behavior proportional to $t$, one must have
\begin{gather*}
\int_\mathbb{R} dt h(t) = -\langle c_{0,1}\rangle_\rho.
\end{gather*}
Substituting in \eqref{3.5}, one arrives at
\begin{gather*}
\chi \int_\mathbb{R} dx |x|f_\mathrm{G}(x) (D t)^{1/2} \simeq - 2\int_0^tds\int_s^\infty du h(u),
\end{gather*}
which implies
\begin{gather}\label{3.15}
h(t) \simeq - c_0 t^{-3/2} ,\qquad c_0 =\tfrac{1}{8} D^{1/2} \chi \int_\mathbb{R} dx |x|f_\mathrm{G}(x).
\end{gather}
The current correlation is negative and decays as $-|t|^{-3/2}$.

With this information, one can now determine the covariance of $J(t)$,
\begin{gather}
\Cov\big(J(t)J(\tau t)\big)
= -\int_0^t \int_0^{\tau t} ds ds'\left(\int_\mathbb{R} du\, h(u) \delta(s-s') - h(s-s')\right) \nonumber\\
\hphantom{\Cov\big(J(t)J(\tau t)\big)}{} = -\int_0^{\tau t}ds \left(2 \int_s^\infty ds' h(s') - \int_{\tau t - s}^{t-s} ds' h(s')\right)\label{3.16}
\end{gather}
with $0 \leq \tau \leq 1$. We insert the asymptotics from \eqref{3.15} in the form $-c_0(c_1+Dt)^{-3/2}$. Then
\begin{gather*}
\Cov\big(J(t)J(\tau t)\big) \simeq \big(1 + \tau^{1/2} - (1 -\tau)^{1/2}\big) (Dt)^{1/2}\chi \int_0^\infty dx x f_\mathrm{G}(x)
\end{gather*}
for large $t$, which one recognizes as the covariance of fractional Brownian motion with Hurst parameter $H = \tfrac{1}{4}$.

\subsection{Non-reversible models, zero propagation speed}\label{sec3.2}

For reversible lattice gases the average current $j(\rho)$ vanishes and a localized perturbation stays centered, compare with \eqref{3.9}. For non-reversible models the average current does not vanish, in general. A small perturbation of the steady state will propagate with velocity $v(\rho) = j'(\rho) $, which generically will be non-zero. The correlator is centered at~$v(\rho) t$. If $v(\rho) \neq 0$, then the sum rule implies that $\Var(J(t)) \sim \sqrt{t}$, indicating that $J(t)$ will be close to a Brownian motion. Fractional Brownian motion can be seen only when the current is integrated along the ray $\{x = v(\rho)t\}$. To properly implement such a notion requires extra considerations, which will be explained in the next subsection. For this part we assume $v(\rho) = 0$. For the ASEP $j(\rho) = (p-q) \rho (1 - \rho)$ and our condition holds only at $\rho = \tfrac{1}{2}$.

Secondly non-reversible models are in the KPZ universality class and the covariance is expected to scale as
\begin{gather}\label{3.21}
S(j,t) \simeq \chi(\Gamma t)^{-2/3} f_\mathrm{KPZ}\big((\Gamma t)^{-2/3}j\big)
\end{gather}
with $\Gamma =\tfrac{1}{2}\chi^2 |j''(\rho)|$ according to KPZ scaling theory~\cite{KMH92}. From the sum rule \eqref{3.5}, again we infer that
\begin{gather*}
\int_\mathbb{R} dt h(t) = -\langle c_{0,1}\rangle_\rho
\end{gather*}
with $h(t)$ given by equation~\eqref{3.20}. Thus, substituting \eqref{3.21}, one arrives at
\begin{gather*}
\chi\int_\mathbb{R} dx |x|f_\mathrm{KPZ}(x) (\Gamma t)^{2/3} \simeq - 2\int_0^tds\int_s^\infty du h(u),
\end{gather*}
which implies
\begin{gather}\label{3.23}
h(t) \simeq - c_0 t^{-4/3} ,\qquad c_0 =\tfrac{1}{9} \Gamma^{2/3} \chi \int_\mathbb{R} dx |x|f_\mathrm{KPZ}(x).
\end{gather}
The current correlation is negative and decays as $-|t|^{-4/3}$. The full covariance is obtained by the same scheme as above, see (\ref{3.16}), with the result
\begin{gather}\label{3.24}
\Cov\big(J(t),J(\tau t)\big) \simeq \tfrac12 \big(1 + \tau^{2/3} - (1 -\tau)^{2/3}\big) (\Gamma t)^{2/3} \chi \int_\R dx |x|f_\mathrm{KPZ}(x)
\end{gather}
valid for large $t$. We recognize the covariance of fractional Brownian motion with Hurst parameter $H = \tfrac{1}{3}$.
Note that the Hurst exponent for the driven lattice gas is larger than the reversible value $\tfrac{1}{4}$. Nevertheless, the process $\RP^{\rm stat}(\tau)$ is not a fractional Brownian motion, since its one-point distribution is known to be non-Gaussian. The non-universal prefactors in~\eqref{eq2.10} and~\eqref{3.24} look dif\/ferent. But they have to agree because of the sum rule~\eqref{3.5}. As explained in Corollary~\ref{CorEquality}, their equivalence can also be verif\/ied directly from the def\/inition.

Our argument is on less secure grounds than in the reversible case. Firstly, the sca\-ling~\eqref{3.21} of the correlator is proved only for the TASEP. Even then, no spectral theorem in the form~\eqref{3.8} is available. But if for TASEP at density $\tfrac{1}{2}$ the current correlator $h(t)$ is assumed to be increasing, then~\eqref{3.24} holds in the limit $t \to \infty$. In Section~\ref{sec4} we display the results of Monte Carlo simulations for the TASEP at density $\tfrac{1}{2}$. They very convincingly conf\/irm $h(t) < 0$, strict increase, and $- t^{-4/3}$ asymptotics, see Figs.~\ref{FigCurrent} and~\ref{FigCurrentLogLog}. For density $\tfrac12$ the theoretically predicted parameters are $\Gamma=\sqrt{2}$ and $c_0=0.02013\ldots$.

\subsection{Non-reversible models, non-zero propagation speed}\label{sec3.3}
We f\/irst have to generalize the sum rule to a current integrated along the ray \mbox{$\{x=vt\}$}, where for notational simplicity we assume $v>0$. As a start-up this will be done for the more transparent case of a continuum stochastic f\/ield $u(x,t)$, which is stationary in time and, for each realization, satisf\/ies the conservation law
\begin{gather}\label{3.25}
\partial_t u(x,t) + \partial_x \mathcal{J}(x,t) = 0.
\end{gather}
The random current f\/ield $\mathcal{J}(x,t)$ is also space-time stationary. Without loss of generality we assume $\E (u(x,t)) = 0$, $\E (\mathcal{J}(x,t) ) = 0$. \eqref{3.25} implies that $(-u,\mathcal{J})$ is a curl-free vector f\/ield on~$\R^2$. Thus there is a potential, resp.\ height function, def\/ined by
\begin{gather}\label{3.26}
h(y,t) = \int_0^t ds \mathcal{J}(0,s) - \int_0^ydx u(x,t),
\end{gather}
where $y \geq 0$ in accordance with $v >0$. $h(y,t)$ does not depend on the choice of the integration path. In particular, one can integrate along the ray $\{x = vt\}$. Then
\begin{gather*}
h(vt,t) = \int_0^t ds \big( \mathcal{J}(vs,s) - vu(vs,s)\big).
\end{gather*}
Along the ray $\{x = vt\}$ the current is given by $s \mapsto \mathcal{J}(vs,s) - vu(vs,s)$, which is a stationary process in $s$ and integrates to $h(vt,t)$.

As before, we def\/ine $S(x,t) = \Cov(u(x,t),u(0,0))$. Then the sum rule~\eqref{3.5} generalises to
\begin{gather}\label{3.28}
\Var(h(y,t)) = \int dx|y-x|S(x,t) - \int dx |x|S(x,0),
\end{gather}
see Appendix~\ref{app1}. If $S(x,t)$ is peaked at $vt$, then the variance of the time-integrated current with end-point $(vt,t)$ ref\/lects the anomalous peak broadening.

For lattice gases the position space is discrete and one has to adjust the scheme. We denote by $J_{j,j+1}([t',t])$ the current across the bond $j,j+1$ integrated over the time-interval $[t',t]$. The height $h(y,t)$, $y \in \Z_+$, is def\/ined in analogy to~\eqref{3.26} as
\begin{gather*}
h(y,t) = J_{0,1}([0,t]) - \sum_{j=1}^y \eta_j(t).
\end{gather*}
The path from $(0,0)$ to $(0,t)$ to $(y,t)$ is deformed into a staircase with step width $1$. Then
\begin{gather*}
h(y,\tfrac{1}{v}y) = \sum_{j=1}^{y} X_j, \quad X_j = J_{j-1,j}\big(\big[\tfrac{1}{v}(j-1), \tfrac{1}{v}j\big]\big) - \eta_j\big(\tfrac{1}{v}j\big).
\end{gather*}
$\{X_j, j \in \Z\}$ is a stationary process and sums up to $h(y,\tfrac{1}{v}y)$.

The sum rule \eqref{3.28} remains valid in the form
\begin{gather*}
\Var( h(y,t)) = \sum_{j \in \mathbb{Z}}|j-y|S(j,t) - \sum_{j \in \mathbb{Z}}|j|S(j,0).
\end{gather*}
 The covariance has the scaling form
\begin{gather}\label{3.32}
S(j,t) \simeq \chi(\Gamma t)^{-2/3} f_\mathrm{KPZ}\big((\Gamma t)^{-2/3}(j -v(\rho)t)\big).
\end{gather}
Now all pieces are assembled. In the def\/inition of $X_j$ we set $v = v(\rho)$. Then the sum rule yields
\begin{gather*}
\sum_{j \in \Z}\Cov(X_0,X_j) = 0 ,
\end{gather*}
and using the scaling form \eqref{3.32} of $S(j,t)$ one arrives at
\begin{gather*}
 \Cov(X_0,X_j) \sim - |j|^{-4/3}
\end{gather*}
for large $|j|$. Then as before one concludes that
\begin{gather*}
\Cov\big(h(v(\rho) t,t),h(v(\rho) \tau t,\tau t)\big)
\simeq \tfrac12\big(1 + \tau^{2/3} - (1 -\tau)^{2/3}\big) (\Gamma t)^{2/3} \chi \int_\R dx |x|f_\mathrm{KPZ}(x)
\end{gather*}
in the scaling regime.

Considering arbitrary space-time rays provides a more complete picture of the current f\/luc\-tua\-tions than merely considering the current across the origin. There is a special direction of slo\-pe~$v(\rho)^{-1}$, along which the covariance is the same as that of fractional Brownian motion with Hurst parameter $H=\tfrac{1}{3}$. For any $v \neq v(\rho)$, the time-integrated current behaves like a Brownian motion.

\section{Numerical simulations}\label{sec4}
To have numerical support of our results we rely on Monte Carlo simulations. As for most of the theory part, we consider the TASEP at density $\tfrac{1}{2}$. From previous works~\cite{FF11} it is known already that the one-point distribution of the rescaled time-integrated current converges quite fast to the asymptotically proven GUE/GOE Tracy--Widom distributions. Thus similar good convergence is expected for the covariance and the current-current correlation.

In the f\/irst set of simulations, we consider the three initial conditions discussed in Section~\ref{sec3} and run the process until time $t_{\max}=10^4$. We measure the vector of the integrated current at the origin $J(\tau t_{\max})$ for $\tau\in\{1/100,2/100,\ldots,99/100,1\}$. We then rescale the current process as (\ref{eq3.1}) and compute numerically the covariance. To facilitate the comparison of the dif\/ferent initial conditions, we divide by the value at $\tau=1$. Therefore in the f\/igures below we plot
\begin{gather*}
\tau\mapsto\Cov(\RP^\diamond(\tau),\RP^\diamond(1))/\Var(\RP^\diamond(1)).
\end{gather*}
Since $t_{\max}=10^4$ is large but not equal to inf\/inity, we computed for comparison the same quantities for $t_{\max}=10^3$ and plotted the numbers with a red dot. For the step and periodic initial conditions we compute the numerical f\/it in the f\/irst and the last~10 of data according to the scaling exponent derived heuristically in Section~\ref{sec3}.

\subsubsection*{Step initial conditions}
For step initial conditions, the number of Monte Carlo trials is $2\times 10^6$ for $t_{\max}=10^3$ and $6\times 10^5$ for $t_{\max}=10^4$. The f\/it functions in Fig.~\ref{FigCovStep} are $\tau\mapsto 0.65 \tau^{2/3}$ and $\tau\mapsto 1-c^{\rm step}(1-\tau)^{2/3}-0.21 (1-\tau)$.
\begin{figure}[t]\centering
\includegraphics[height=5.5cm]{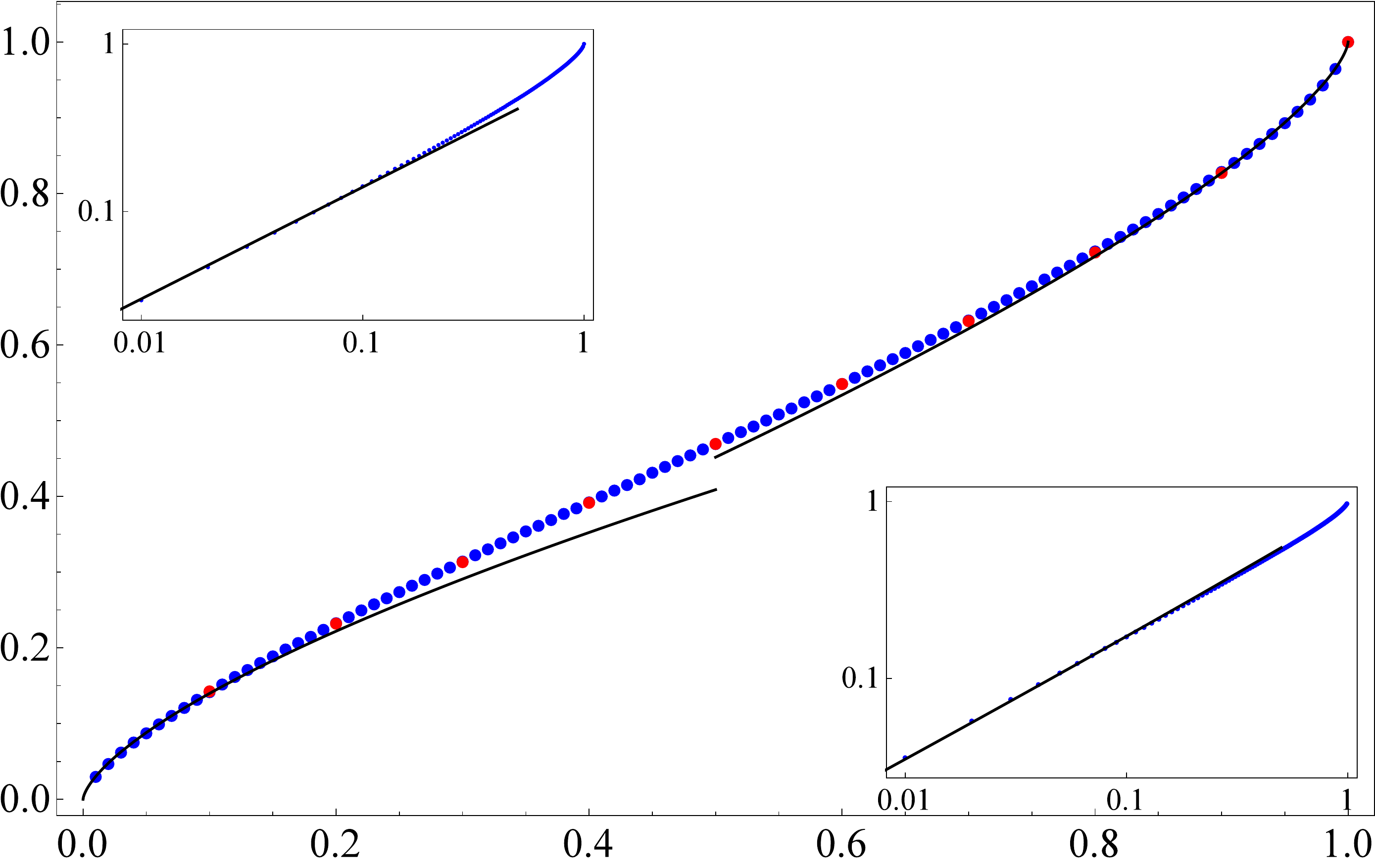}
\caption{Plot of $\tau\mapsto\Cov(\RP^{\rm step}(\tau),\RP^{\rm step}(1))/\Var(\RP^{\rm step}(1))$. The top-left (resp.\ right-bottom) inset is the log-log plot around $\tau=0$ (resp.\ $\tau=1$).}\label{FigCovStep}
\end{figure}

\subsubsection*{Periodic initial conditions}
For periodic initial conditions, the number of Monte Carlo trials is $10^6$ for \mbox{$t_{\max}=10^3$} and $4\times 10^5$ for $t_{\max}=10^4$. The f\/it functions in Fig.~\ref{FigCovFlat} are $\tau\mapsto 0.97 \tau^{4/3}$ and $\tau\mapsto 1-c^{\rm f\/lat}(1-\tau)^{2/3}-0.23 (1-\tau)$.
\begin{figure}[t]\centering
\includegraphics[height=5.5cm]{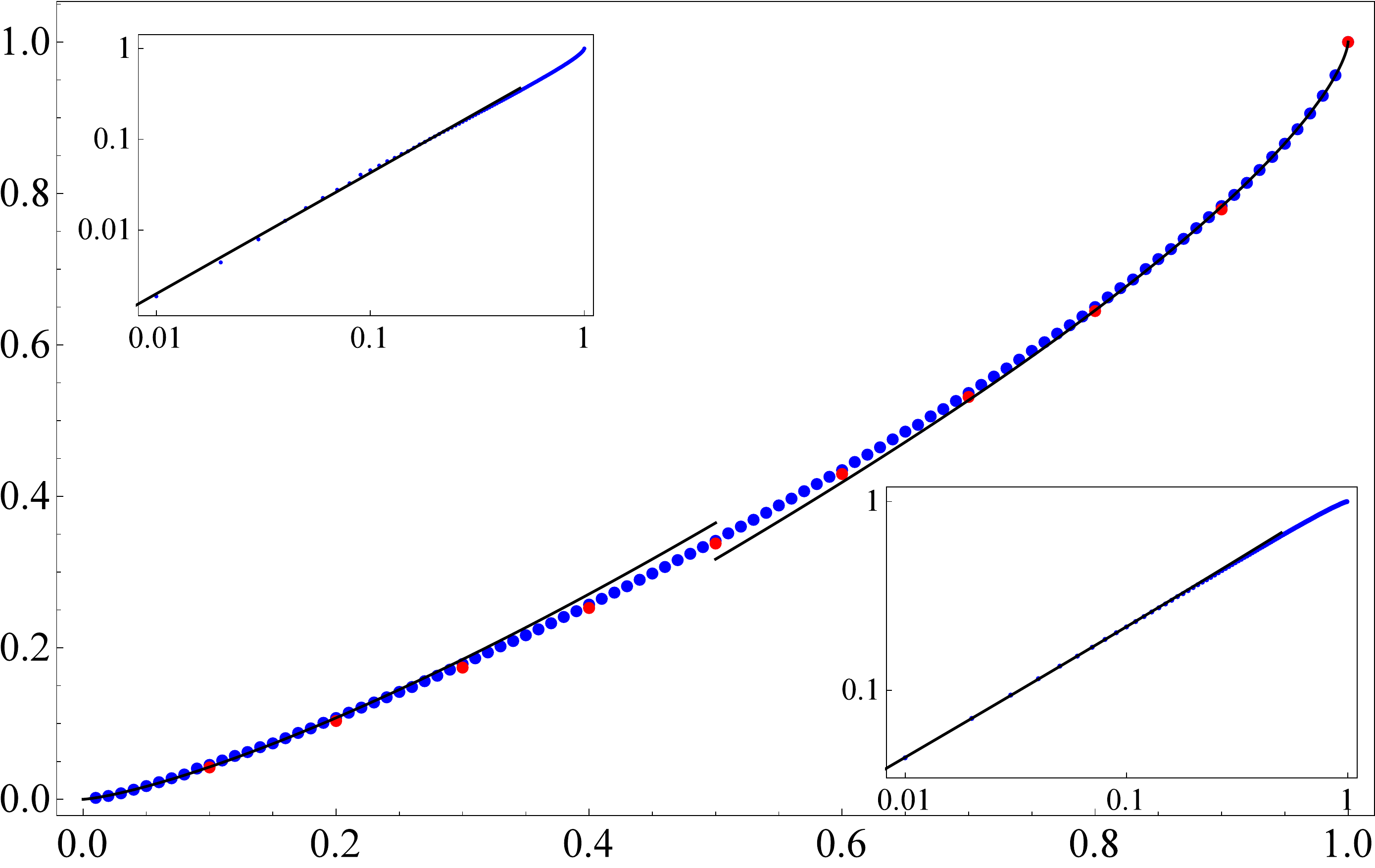}
\caption{Plot of $\tau\mapsto\Cov(\RP^{\rm f\/lat}(\tau),\RP^{\rm f\/lat}(1))/\Var(\RP^{\rm f\/lat}(1))$. The top-left (resp.\ right-bottom) inset is the log-log plot around $\tau=0$ (resp.\ $\tau=1$).}\label{FigCovFlat}
\end{figure}

\subsubsection*{Stationary initial conditions}
For periodic initial conditions, the number of Monte Carlo trials is $3\times 10^5$ for $t_{\max}=10^3$ and~$10^5$ for $t_{\max}=10^4$. The f\/it functions in Fig.~\ref{FigCovStat} is obtained from (\ref{3.24}) by normalization, namely $\tau\mapsto \frac12 (1+\tau^{2/3}-(1-\tau)^{2/3})$.
\begin{figure}[t]\centering
\includegraphics[height=5.5cm]{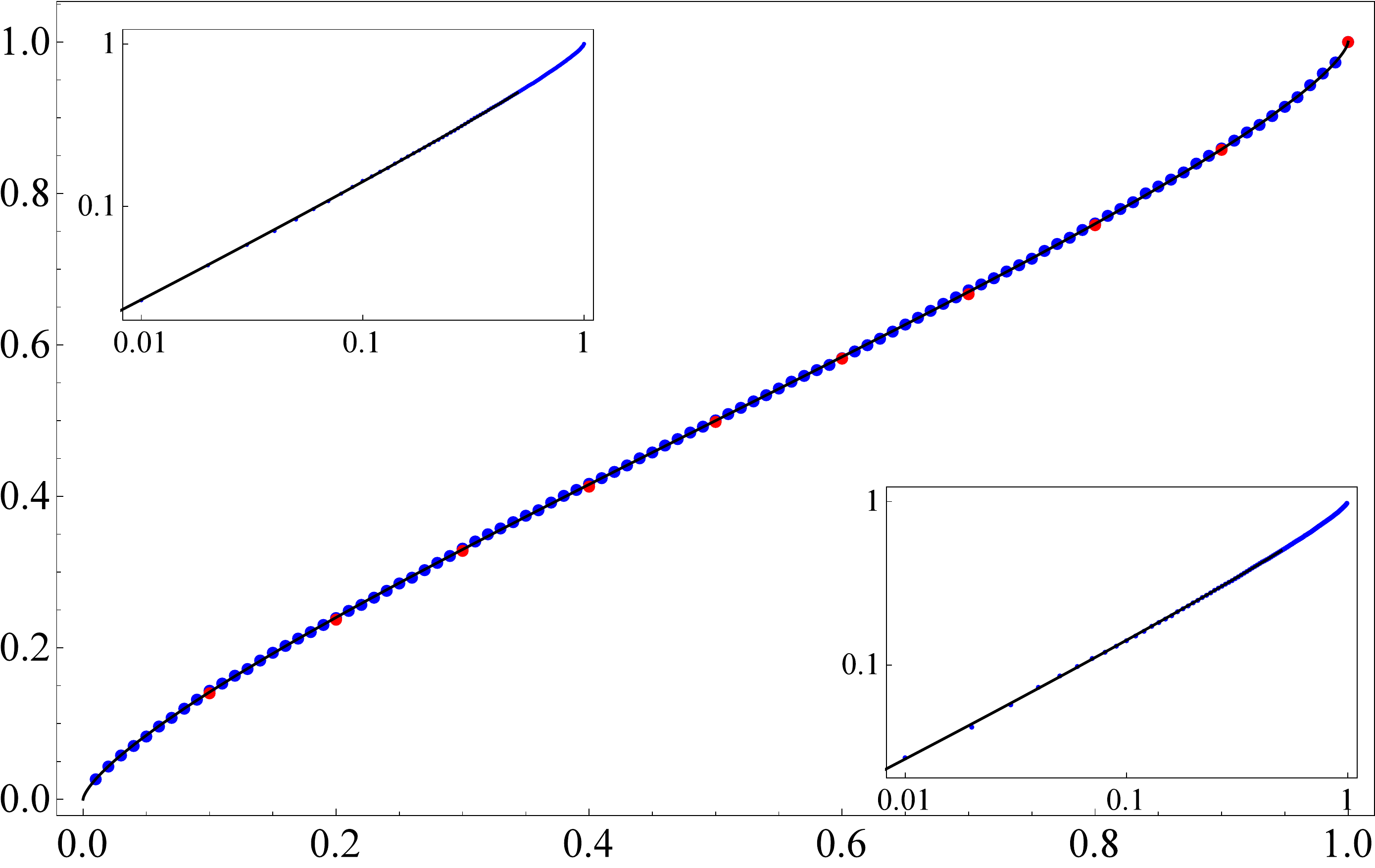}
\caption{Plot of $\tau\mapsto\Cov(\RP^{\rm stat}(\tau),\RP^{\rm stat}(1))/\Var(\RP^{\rm stat}(1))$. The top-left inset is the log-log plot around $\tau=0$ and the right-bottom inset is the log-log plot around $\tau=1$. The f\/it is made with the function $\tau\mapsto \frac12 (1+\tau^{2/3}-(1-\tau)^{2/3})$.}
\label{FigCovStat}
\end{figure}

For the stationary initial conditions, we also simulated the current-current correlations. To measure its smooth part $h(t)$, def\/ined in (\ref{3.20}), the TASEP is run up to time $t=50$ with $50\times 10^6$ Monte Carlo trials. The results are displayed in Figs.~\ref{FigCurrent} and~\ref{FigCurrentLogLog}. The predicted power law of $t^{-4/3}$, including its prefactor, is convincingly conf\/irmed.
\begin{figure}[t]\centering
\includegraphics[height=5.2cm]{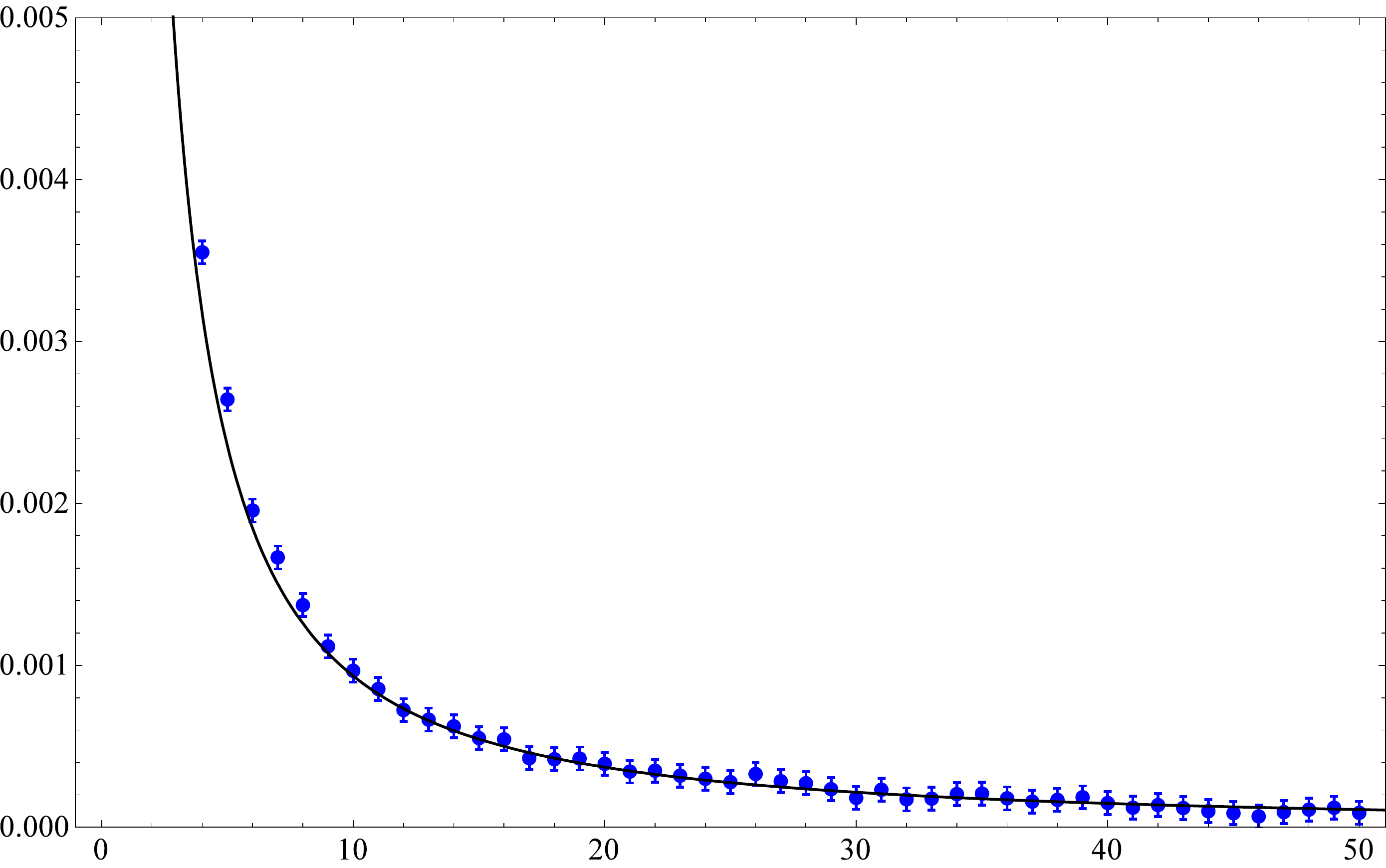}
\caption{The smooth part of the current-current correlations for TASEP. We plot $-h(t)$ and the theoretical large time behavior (\ref{3.23}), namely $0.02013 \cdot t^{-4/3}$.}\label{FigCurrent}
\end{figure}

\begin{figure}[t!]\centering
\includegraphics[height=5.5cm]{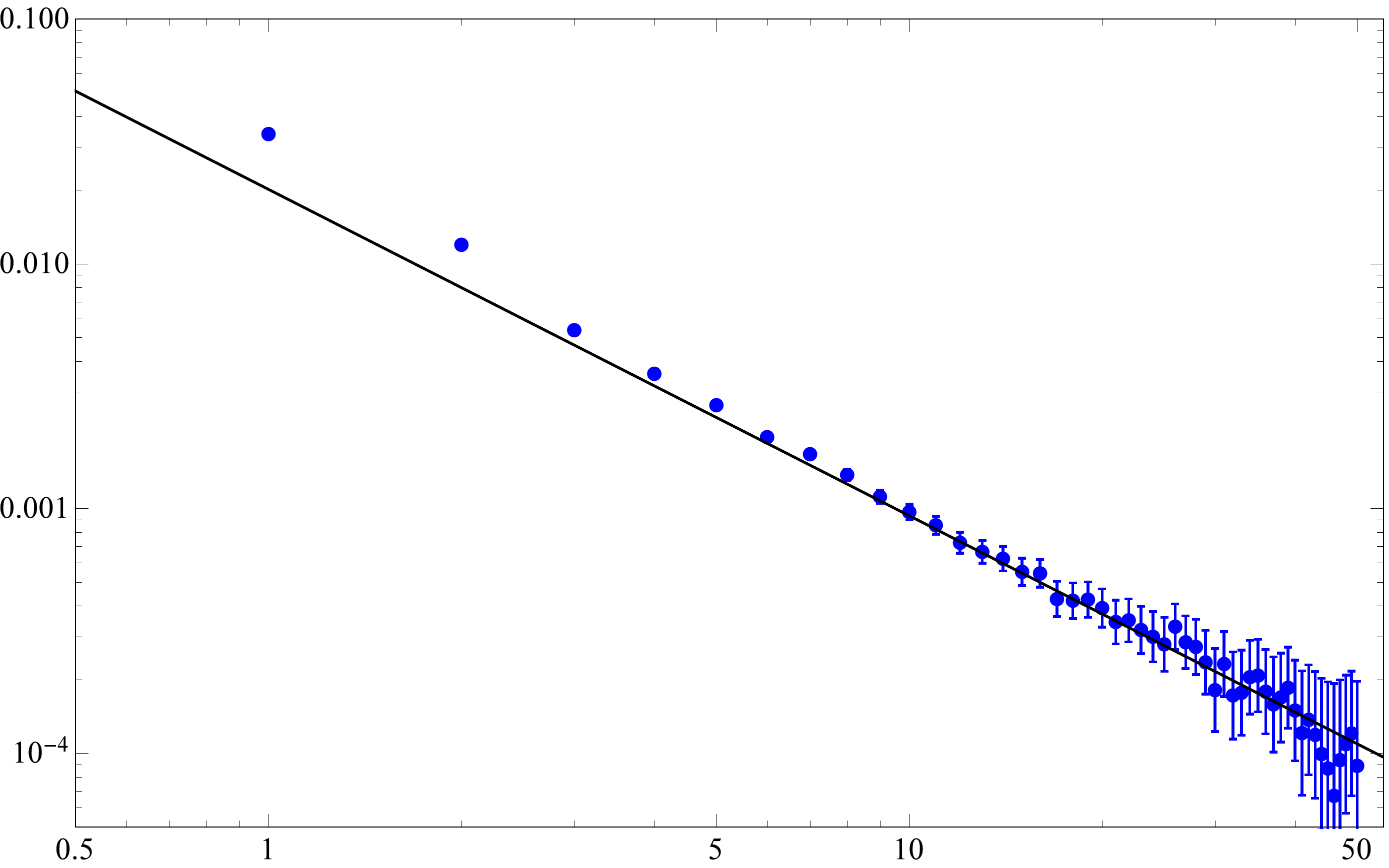}
\caption{Log-log plot of the smooth part of the current-current correlation for TASEP.}\label{FigCurrentLogLog}
\end{figure}

\appendix
\section{Scaling functions and limiting distributions}\label{app}
We recall the def\/initions of the GUE/GOE Tracy--Widom and the Baik--Rains distribution functions as well as the scaling function $f_{\rm KPZ}$ used for the two-point function.

\begin{defin}
The GUE Tracy--Widom distribution function is def\/ined by{\samepage
\begin{gather*}
F_{\rm GUE}(s) =\det(\Id-K_{2,s})_{L^2(\R_+)} =\sum_{n\geq 0}\frac{(-1)^n}{n!} \int_{\R_+} dx_1\cdots \int_{\R_+} dx_n \det\left[K_{2,s}(x_i,x_j)\right]_{1\leq i,j\leq n}
\end{gather*}
with the kernel $K_{2,s}(x,y)=\int_{\R_+} d\lambda \Ai(x+s+\lambda)\Ai(y+s+\lambda)$.}
\end{defin}

\begin{defin}
The GOE Tracy--Widom distribution function is def\/ined by{\samepage
\begin{gather*}
F_{\rm GOE}(s) =\det(\Id-K_1)_{L^2(\R_+)} =\sum_{n\geq 0}\frac{(-1)^n}{n!} \int_{\R_+} dx_1\cdots \int_{\R_+} dx_n \det\left[K_{1,s}(x_i,x_j)\right]_{1\leq i,j\leq n}
\end{gather*}
with the kernel $K_{1,s}(x,y)=\Ai(x+y+s)$.}
\end{defin}

\begin{defin}
The Baik--Rains distribution function is def\/ined by
\begin{gather*}
F_{\rm BR}(s)=\frac{\partial}{\partial s} (F_{\rm GUE}(s) g(s)),
\end{gather*}
where $g(s)$ is given as follows
\begin{gather*}
g(s)=s+\int_{\R_+^2} dx dy \Ai(x+y+s)-\int_{\R_+^2} dx dy \Phi_s(x) (\Id-P_0 K_{2,s} P_0)^{-1}(x,y)\Psi_s(y),
\end{gather*}
where $P_s$ is the projection onto $(s,\infty)$ and
\begin{gather*}
\Phi_s(x)=\int_{\R_+} dy K_{2,s}(x,y),\qquad \Psi_s(y)=1-\int_{\R_+} dx \Ai(x+y+s).
\end{gather*}
\end{defin}

\begin{defin}
The KPZ scaling function $f_{\rm KPZ}$ is def\/ined by
\begin{gather*}
f_{\rm KPZ}(w)=\frac14 \frac{\partial^2}{\partial w^2} \int_\R s^2 dF_w(s),
\end{gather*}
where
\begin{gather*}
F_w(s)=\frac{\partial}{\partial s} \big(F_{\rm GUE}\big(s+w^2\big) g\big(s+w^2,w\big)\big),
\end{gather*}
where $g(s,w)$ is given by
\begin{gather*}
g(s,w)=e^{-w^3/3}\left(\int_{\R_-}dx dy e^{w(x+y)}\Ai(x+y+s)+\int_{\R_+^2}dx dy \Phi_{w,s}(x)\rho_s(x,y)\Psi_{w,s}(y)\right).
\end{gather*}
Here, $\rho_s(x,y)=(\Id-P_0 K_{2,s} P_0)^{-1}(x,y)$, and
\begin{gather*}
\Phi_{w,s}(x)=\int_{\R_-} dz e^{w (z+s)}K_{2,s}(z,x),\qquad \Psi_{w,s}(y)=\int_{\R_-}dz e^{wz}\Ai(y+z+s).
\end{gather*}
\end{defin}
The scaling function $f_{\rm KPZ}(w)$ is even with $\int_\R dw f_{\rm KPZ}(w) |w|=0.287599\ldots$. Remark that $g(s,0)=g(s)$.

Here is another identity that allows us to compare the two formulas obtained for the stationary case.
\begin{lem}\label{LemmaStatTwoFormulas}
It holds
\begin{gather*}
\int_\R dx |x| f_{\rm KPZ}(x)=\frac12\Var(\xi_{\rm BR}).
\end{gather*}
\end{lem}
\begin{proof}
Using the above def\/initions and the fact that $f_{\rm KPZ}$ is an even function, we have
\begin{gather*}
\int_\R dx |x| f_{\rm KPZ}(x) = \frac12\int_0^\infty dx x \frac{\partial^2}{\partial x^2} \int_\R ds s^2 \frac{\partial^2}{\partial s^2} \big(F_{\rm GUE}\big(s+x^2\big) g\big(s+x^2,x\big)\big)\\
\hphantom{\int_\R dx |x| f_{\rm KPZ}(x)}{}
=\frac12\int_\R ds s^2 \frac{\partial^2}{\partial s^2} \int_0^\infty dx x \frac{\partial^2}{\partial x^2}\big(F_{\rm GUE}\big(s+x^2\big) g\big(s+x^2,x\big)\big)\\
\hphantom{\int_\R dx |x| f_{\rm KPZ}(x)}{}
=\frac12\int_\R ds s^2 \frac{\partial^2}{\partial s^2} F_{\rm GUE}(s) g(s,0) = \frac12 \Var(\xi_{\rm BR}),
\end{gather*}
where in the third step we use integration by parts twice and in the last step the fact that $\E(\xi_{\rm BR})=0$.
\end{proof}

\begin{cor}\label{CorEquality}
 For the stationary TASEP, the prefactors in \eqref{eq2.10} and \eqref{3.24} are identical.
\end{cor}
\begin{proof}
For TASEP with stationary initial conditions and density $1/2$, the parameters in (\ref{3.24}) are $\chi=1/4$ and $\Gamma=\sqrt{2}$. Rescaling $h(0,t)=J(t)$ as in (\ref{eq3.1}) we get from (\ref{3.24})
\begin{gather*}
\Cov\big(\RP^{\rm stat}(\tau),\RP^{\rm stat}(1)\big) =\frac12\big(1+\tau^{2/3}-(1-\tau)^{2/3}\big) 2^{8/3}\Gamma^{2/3}\chi \int_\R dx |x| f_{\rm KPZ}(x)\\
\hphantom{\Cov\big(\RP^{\rm stat}(\tau),\RP^{\rm stat}(1)\big)}{} =\frac12\big(1+\tau^{2/3}-(1-\tau)^{2/3}\big) \Var(\xi_{\rm BR}) = (\ref{eq2.10}),
\end{gather*}
where we used Lemma~\ref{LemmaStatTwoFormulas} in the second equality.
\end{proof}

\section{Sum rule and current-current correlations}\label{app1}
\subsection{The sum rule} We prove that
\begin{gather*}
\Var( h(y,t)) = \sum_{j \in \mathbb{Z}}|j-y|S(j,t) - \sum_{j \in \mathbb{Z}}|j|S(j,0).
\end{gather*}
We use the def\/inition \eqref{3.26}. Expanding out the square yields the terms (I), (II), and (III). Let us introduce the short hand
$\sum\limits_{j \in \Z} g_j \eta_j(t) = \eta(g,t)$ and correspondingly $\sum\limits_{j \in \Z} f_j J_{j,j+1}(t) = J(f,t)$. By the conservation law,{\samepage
\begin{gather*}
\Cov(\eta(f,t),\eta(g,t)) = \Cov\big(J\big(\partial^\mathrm{T}\partial f,t\big), J(g,t)\big)
\end{gather*}
with $(\partial f)_j = f_{j+1} -f_j$.}

Choosing $g_j = \delta_{0j}$ and $f_j = |j|$, hence $ (\partial^\mathrm{T}\partial f)_j= -2\delta_{0j}$, one arrives at
\begin{gather*}
\mathrm{(I)} = \Var(J_{0,1}(t)) = \tfrac{1}{2} \sum_{j \in \Z} |j|\big(S(j,t) + S(-j,t) - 2 S(j,0)\big),
\end{gather*}
where we used stationarity in $j$. Next we consider the cross term starting from
\begin{gather*}
 \Cov(\eta(f,t),\eta(g,0)) = - \Cov\big( J(f,t),\eta\big(\partial^\mathrm{T} g,0\big)\big).
\end{gather*}
Choosing $f_j = \delta_{jy}$ and $-\partial^\mathrm{T} g$ as the indicator function of $[1,\dots,y]$ yields
\begin{gather*}
\mathrm{(II)} = - 2 \Cov (J_{y,y+1}(t),\eta(\partial^\mathrm{T} g,0)) = 2 \sum_{j \in \Z}g_j\big(S(y-j,t) - S(y-j,0)\big).
\end{gather*}
Finally
\begin{gather*}
\mathrm{(III)} = \sum_{j=1}^y \sum_{i=1}^yS(j-i,0).
\end{gather*}
Summing all three terms establishes the claim.

\subsection{Current-current correlation}
One can think of $\mathsf{j}_{j,j+1}(t)$ as a point process with weights $\pm1$. Then the covariance has a self-part, proportional to $\delta(t - t')$, and a continuous part. Such a decomposition holds also for current correlations. We f\/irst consider the self-part and introduce the short hands $q_j = \eta_j(1- \eta_{j+1})$, $\bar{q}_j = (1- \eta_j)\eta_{j+1}$, $q_j(\eta(t)) = q_j(t)$, $J_{j,j+1}([s,t]) = J_{j,j+1}(t)-J_{j,j+1}(s)$ for $0\leq s < t$. Note that $q_j - \bar{q}_j = \eta_j - \eta_{j+1}$. Then
\begin{gather*}
\lim_{\delta \to 0} \delta^{-1}
\E \big(J_{i,i+1}([t,t+\delta])J_{j,j+1}([t,t+\delta])\big) \\
\qquad{} = \lim_{\delta \to 0} \delta^{-1} \E \big((q_i(t) \bar{q}_i(t+ \delta) - \bar{q}_i(t)q_i(t+ \delta))(q_j(t) \bar{q}_j(t+ \delta) - \bar{q}_j(t)q_j(t+ \delta))\big)\\
\qquad{} = \langle q_iq_jL \bar{q}_i\bar{q}_j +\bar{q}_i\bar{q}_j L q_iq_j- q_i\bar{q}_jL \bar{q}_iq_j-
\bar{q}_iq_jL q_i\bar{q}_j\rangle_\rho = \delta_{i,j} \langle c_{j,j+1} \rangle_\rho.
\end{gather*}
By the same method the continuous part is obtained as
\begin{gather*}
\lim_{\delta \to 0} \delta^{-2}\E \big(J_{i,i+1}([s,s+\delta]J_{j,j+1}([t,t+\delta])\big)\\
\qquad{} =\lim_{\delta \to 0} \delta^{-1} \E \big((q_i(s) \bar{q}_i(s+ \delta) - \bar{q}_i(s)q_i(s+ \delta)) \mathrm{e}^{L(t - s -\delta)}r_{j,j+1}(\eta(s+\delta))\big) \\
\qquad{} = \big\langle(\eta_i - \eta_{i+1})c_{i,i+1}(\mathrm{e}^{L(t - s)}r_{j,j+1})(\eta^{j,j+1})\big\rangle_\rho
= -\big\langle r^\mathrm{R}_{i,i+1}\mathrm{e}^{L(t - s)}r_{j,j+1}\big\rangle_\rho.
\end{gather*}
Here $r^\mathrm{R}_{i,i+1}$ is the reverse current def\/ined by
\begin{gather*}
r^\mathrm{R}_{j,j+1}(\eta) = \frac{\mu_s(\eta^{j,j+1})}{\mu_s(\eta)} c_{j,j+1}\big(\eta^{j,j+1}\big)(\eta_j - \eta_{j+1}).
\end{gather*}
Hence
\begin{gather*}
\E \big(\mathsf{j}_{i,i+1}(s)\mathsf{j}_{j,j+1}(t)\big) = \langle c_{j,j+1}\rangle_\rho\delta _{i,j}\delta(s-t)
- \big\langle r^\mathrm{R}_{i,i+1}\mathrm{e}^{L(t - s)}r_{j,j+1}\big\rangle_\rho.
\end{gather*}
In particular,
\begin{gather*}
\Cov\big(\mathsf{j}(s),\mathsf{j}(t)\big)=\E \big(\mathsf{j}(s) \mathsf{j}(t)\big) - j(\rho)^2 = \langle c_{0,1}\rangle_\rho\delta(s-t)+ h(s-t),
\end{gather*}
where
\begin{gather*}
h(t) = -\big\langle \big(r^\mathrm{R}_{0,1} -j(\rho)\big)\mathrm{e}^{L|t|}(r_{0,1}- j(\rho))\big\rangle_\rho.
\end{gather*}

\subsection*{Acknowledgements}
The work of P.L.~Ferrari is supported by the German Research Foundation via the SFB 1060--B04 project. The f\/inal version of our contribution was written when both of us visited in early 2016 the Kavli Institute of Theoretical Physics at Santa Barbara. The research stay of H.~Spohn at KITP is supported by the Simons Foundation. This research was supported in part by the National Science Foundation under Grant No.~NSF PHY11-25915. We thank Kazumasa Takeuchi for illuminating discussions on the comparison with his experimental results and Joachim Krug for explaining to us earlier work on time correlations.

\pdfbookmark[1]{References}{ref}
\LastPageEnding


\begin{thebibliography}{99}
\footnotesize\itemsep=0pt

\bibitem{BFP09}
Baik J., Ferrari P.L., P{\'e}ch{\'e} S., Limit process of stationary {TASEP}
 near the characteristic line, \href{http://dx.doi.org/10.1002/cpa.20316}{\textit{Comm. Pure Appl. Math.}} \textbf{63}
 (2010), 1017--1070, \href{http://arxiv.org/abs/0907.0226}{arXiv:0907.0226}.

\bibitem{BKS12}
Baik J., Liechty K., Schehr G., On the joint distribution of the maximum and
 its position of the {${\rm Airy}_2$} process minus a parabola,
 \href{http://dx.doi.org/10.1063/1.4746694}{\textit{J.~Math. Phys.}} \textbf{53} (2012), 083303, 13~pages,
 \href{http://arxiv.org/abs/1205.3665}{arXiv:1205.3665}.

\bibitem{BR00}
Baik J., Rains E.M., Limiting distributions for a polynuclear growth model with
 external sources, \href{http://dx.doi.org/10.1023/A:1018615306992}{\textit{J.~Stat. Phys.}} \textbf{100} (2000), 523--541,
 \href{http://arxiv.org/abs/math.PR/0003130}{math.PR/0003130}.

\bibitem{BR99b}
Baik J., Rains E.M., The asymptotics of monotone subsequences of involutions,
 \href{http://dx.doi.org/10.1215/S0012-7094-01-10921-6}{\textit{Duke Math.~J.}} \textbf{109} (2001), 205--281,
 \href{http://arxiv.org/abs/math.CO/9905084}{math.CO/9905084}.

\bibitem{BSS14}
Basu R., Sidoravicius V., Sly A., Last passage percolation with a defect line
 and the solution of the slow bond problem, \href{http://arxiv.org/abs/1408.3464}{arXiv:1408.3464}.

\bibitem{BFP08}
Bornemann F., Ferrari P.L., Pr{\"a}hofer M., The {${\rm Airy}_1$} process is
 not the limit of the largest eigenvalue in {GOE} matrix dif\/fusion,
 \href{http://dx.doi.org/10.1007/s10955-008-9621-0}{\textit{J.~Stat. Phys.}} \textbf{133} (2008), 405--415, \href{http://arxiv.org/abs/0806.3410}{arXiv:0806.3410}.

\bibitem{BCFV14}
Borodin A., Corwin I., Ferrari P., Vet{\H{o}} B., Height f\/luctuations for the
 stationary {KPZ} equation, \href{http://dx.doi.org/10.1007/s11040-015-9189-2}{\textit{Math. Phys. Anal. Geom.}} \textbf{18}
 (2015), Art.~20, 95~pages, \href{http://arxiv.org/abs/1407.6977}{arXiv:1407.6977}.

\bibitem{BF07}
Borodin A., Ferrari P.L., Large time asymptotics of growth models on space-like
 paths. {I}.~{P}ush{ASEP}, \href{http://dx.doi.org/10.1214/EJP.v13-541}{\textit{Electron.~J. Probab.}} \textbf{13} (2008),
 no.~50, 1380--1418, \href{http://arxiv.org/abs/0707.2813}{arXiv:0707.2813}.

\bibitem{BFPS06}
Borodin A., Ferrari P.L., Pr{\"a}hofer M., Sasamoto T., Fluctuation properties
 of the {TASEP} with periodic initial conf\/iguration, \href{http://dx.doi.org/10.1007/s10955-007-9383-0}{\textit{J.~Stat. Phys.}}
 \textbf{129} (2007), 1055--1080, \href{http://arxiv.org/abs/math-ph/0608056}{math-ph/0608056}.

\bibitem{BG12}
Borodin A., Gorin V., Lectures on integrable probability, \href{http://arxiv.org/abs/1212.3351}{arXiv:1212.3351}.

\bibitem{Bur56}
Burke P.J., The output of a queuing system, \href{http://dx.doi.org/10.1287/opre.4.6.699}{\textit{Operations Res.}} \textbf{4}
 (1956), 699--704.

\bibitem{Cha94}
Chang C.-C., Equilibrium f\/luctuations of gradient reversible particle systems,
 \href{http://dx.doi.org/10.1007/BF01193701}{\textit{Probab. Theory Related Fields}} \textbf{100} (1994), 269--283.

\bibitem{Cha96}
Chang C.-C., Equilibrium f\/luctuations of nongradient reversible particle
 systems, in Nonlinear Stochastic {PDE}s ({M}inneapolis, {MN}, 1994),
 \href{http://dx.doi.org/10.1007/978-1-4613-8468-7_2}{\textit{IMA Vol. Math. Appl.}}, Vol.~77, Springer, New York, 1996, 41--51.

\bibitem{Cor11}
Corwin I., The {K}ardar--{P}arisi--{Z}hang equation and universality class,
 \href{http://dx.doi.org/10.1142/S2010326311300014}{\textit{Random Matrices Theory Appl.}} \textbf{1} (2012), 1130001, 76~pages,
 \href{http://arxiv.org/abs/1106.1596}{arXiv:1106.1596}.

\bibitem{CFP10a}
Corwin I., Ferrari P.L., P{\'e}ch{\'e} S., Limit processes for {TASEP} with
 shocks and rarefaction fans, \href{http://dx.doi.org/10.1007/s10955-010-9995-7}{\textit{J.~Stat. Phys.}} \textbf{140} (2010),
 232--267, \href{http://arxiv.org/abs/1002.3476}{arXiv:1002.3476}.

\bibitem{CFP10b}
Corwin I., Ferrari P.L., P{\'e}ch{\'e} S., Universality of slow decorrelation
 in {KPZ} growth, \href{http://dx.doi.org/10.1214/11-AIHP440}{\textit{Ann. Inst. Henri Poincar\'e Probab. Stat.}}
 \textbf{48} (2012), 134--150, \href{http://arxiv.org/abs/1001.5345}{arXiv:1001.5345}.

\bibitem{CH11}
Corwin I., Hammond A., Brownian {G}ibbs property for {A}iry line ensembles,
 \href{http://dx.doi.org/10.1007/s00222-013-0462-3}{\textit{Invent. Math.}} \textbf{195} (2014), 441--508, \href{http://arxiv.org/abs/1108.2291}{arXiv:1108.2291}.

\bibitem{CH16}
Corwin I., Hammond A., {P}rivate communication, 2016.

\bibitem{DMF02}
De~Masi A., Ferrari P.A., Flux f\/luctuations in the one dimensional nearest
 neighbors symmetric simple exclusion process, \href{http://dx.doi.org/10.1023/A:1014577928229}{\textit{J.~Stat. Phys.}}
 \textbf{107} (2002), 677--683, \href{http://arxiv.org/abs/math.PR/0103233}{math.PR/0103233}.

\bibitem{Dot13}
Dotsenko V., Two-time free energy distribution function in {$(1+1)$} directed
 polymers, \href{http://dx.doi.org/10.1088/1742-5468/2013/06/P06017}{\textit{J.~Stat. Mech. Theory Exp.}} \textbf{2013} (2013), P06017,
 23~pages, \href{http://arxiv.org/abs/1304.0626}{arXiv:1304.0626}.

\bibitem{Dot16}
Dotsenko V., On two-time distribution functions in $(1+1)$ random directed
 polymers, \href{http://dx.doi.org/10.1088/1751-8113/49/27/27LT01}{\textit{J.~Phys.~A: Math. Theor.}} \textbf{49} (2016), 27LT01,
 8~pages, \href{http://arxiv.org/abs/1603.08945}{arXiv:1603.08945}.

\bibitem{Fer08}
Ferrari P.L., Slow decorrelations in {K}ardar--{P}arisi--{Z}hang growth,
 \href{http://dx.doi.org/10.1088/1742-5468/2008/07/P07022}{\textit{J.~Stat. Mech. Theory Exp.}} \textbf{2008} (2008), P07022, 18~pages,
 \href{http://arxiv.org/abs/0806.1350}{arXiv:0806.1350}.

\bibitem{Fer07}
Ferrari P.L., The universal {${\rm Airy}_1$} and {${\rm Airy}_2$} processes in
 the totally asymmetric simple exclusion process, in Integrable Systems and
 Random Matrices: in Honor of Percy Deift, \href{http://dx.doi.org/10.1090/conm/458/08944}{\textit{Contemp. Math.}}, Vol.~458,
 Editors J.~Baik, T.~Kriecherbauer, L.-C.~Li, K.D.T.-R.~McLaughlin, C.~Tomei,
 Amer. Math. Soc., Providence, RI, 2008, 321--332, \href{http://arxiv.org/abs/math-ph/0701021}{math-ph/0701021}.

\bibitem{FF11}
Ferrari P.L., Frings R., Finite time corrections in {KPZ} growth models,
 \href{http://dx.doi.org/10.1007/s10955-011-0318-4}{\textit{J.~Stat. Phys.}} \textbf{144} (2011), 1123--1150, \href{http://arxiv.org/abs/1104.2129}{arXiv:1104.2129}.

\bibitem{FS05a}
Ferrari P.L., Spohn H., Scaling limit for the space-time covariance of the
 stationary totally asymmetric simple exclusion process, \href{http://dx.doi.org/10.1007/s00220-006-1549-0}{\textit{Comm. Math.
 Phys.}} \textbf{265} (2006), 1--44, \href{http://arxiv.org/abs/math-ph/0504041}{math-ph/0504041}.

\bibitem{FS10}
Ferrari P.L., Spohn H., Random growth models, in The {O}xford Handbook of
 Random Matrix Theory, Editors G.~Akemann, J.~Baik, P.~Di~Francesco, Oxford
 University Press, Oxford, 2011, 782--801, \href{http://arxiv.org/abs/1003.0881}{arXiv:1003.0881}.

\bibitem{FSW15}
Ferrari P.L., Spohn H., Weiss T., Brownian motions with one-sided collisions:
 the stationary case, \href{http://dx.doi.org/10.1214/EJP.v20-4177}{\textit{Electron.~J. Probab.}} \textbf{20} (2015),
 no.~69, 41~pages, \href{http://arxiv.org/abs/1502.01468}{arXiv:1502.01468}.

\bibitem{Ha07}
H{\"a}gg J., Local {G}aussian f\/luctuations in the {A}iry and discrete {PNG}
 processes, \href{http://dx.doi.org/10.1214/07-AOP353}{\textit{Ann. Probab.}} \textbf{36} (2008), 1059--1092,
 \href{http://arxiv.org/abs/math.PR/0701880}{math.PR/0701880}.

\bibitem{SI04b}
Imamura T., Sasamoto T., Polynuclear growth model with external source and
 random matrix model with deterministic source, \href{http://dx.doi.org/10.1103/PhysRevE.71.041606}{\textit{Phys. Rev.~E}}
 \textbf{71} (2005), 041606, 12~pages, \href{http://arxiv.org/abs/math-ph/0411057}{math-ph/0411057}.

\bibitem{Jo00b}
Johansson K., Shape f\/luctuations and random matrices, \href{http://dx.doi.org/10.1007/s002200050027}{\textit{Comm. Math.
 Phys.}} \textbf{209} (2000), 437--476, \href{http://arxiv.org/abs/math.CO/9903134}{math.CO/9903134}.

\bibitem{Jo00}
Johansson K., Transversal f\/luctuations for increasing subsequences on the
 plane, \href{http://dx.doi.org/10.1007/s004400050258}{\textit{Probab. Theory Related Fields}} \textbf{116} (2000), 445--456,
 \href{http://arxiv.org/abs/math.PR/9910146}{math.PR/9910146}.

\bibitem{Jo03b}
Johansson K., Discrete polynuclear growth and determinantal processes,
 \href{http://dx.doi.org/10.1007/s00220-003-0945-y}{\textit{Comm. Math. Phys.}} \textbf{242} (2003), 277--329,
 \href{http://arxiv.org/abs/math.PR/0206208}{math.PR/0206208}.

\bibitem{Joh15}
Johansson K., Two time distribution in Brownian directed percolation,
 \href{http://dx.doi.org/10.1007/s00220-016-2660-5}{\textit{Comm. Math. Phys.}}, {t}o appear, \href{http://arxiv.org/abs/1502.00941}{arXiv:1502.00941}.

\bibitem{KK99}
Kallabis H., Krug J., Persistence of {K}ardar--{P}arisi--{Z}hang interfaces,
 \href{http://dx.doi.org/10.1209/epl/i1999-00125-0}{\textit{Europhys. Lett.}} \textbf{45} (1999), 20--25,
 \href{http://arxiv.org/abs/cond-mat/9809241}{cond-mat/9809241}.

\bibitem{KPZ86}
Kardar M., Parisi G., Zhang Y.-C., Dynamic scaling of growing interfaces,
 \href{http://dx.doi.org/10.1103/PhysRevLett.56.889}{\textit{Phys. Rev. Lett.}} \textbf{56} (1986), 889--892.

\bibitem{KMH92}
Krug J., Meakin P., Halpin-Healy T., Amplitude universality for driven
 interfaces and directed polymers in random media, \href{http://dx.doi.org/10.1103/PhysRevA.45.638}{\textit{Phys. Rev.~A}}
 \textbf{45} (1992), 638--653.

\bibitem{MFQR11}
Moreno~Flores G., Quastel J., Remenik D., Endpoint distribution of directed
 polymers in {$1+1$} dimensions, \href{http://dx.doi.org/10.1007/s00220-012-1583-z}{\textit{Comm. Math. Phys.}} \textbf{317}
 (2013), 363--380, \href{http://arxiv.org/abs/1106.2716}{arXiv:1106.2716}.

\bibitem{PS08}
Peligrad M., Sethuraman S., On fractional {B}rownian motion limits in one
 dimensional nearest-neighbor symmetric simple exclusion, \textit{ALEA Lat.
 Am.~J. Probab. Math. Stat.} \textbf{4} (2008), 245--255, \href{http://arxiv.org/abs/0711.0017}{arXiv:0711.0017}.

\bibitem{PS02}
Pr{\"a}hofer M., Spohn H., Scale invariance of the {PNG} droplet and the {A}iry
 process, \href{http://dx.doi.org/10.1023/A:1019791415147}{\textit{J.~Stat. Phys.}} \textbf{108} (2002), 1071--1106,
 \href{http://arxiv.org/abs/math.PR/0105240}{math.PR/0105240}.

\bibitem{PS02b}
Pr{\"a}hofer M., Spohn H., Exact scaling functions for one-dimensional
 stationary {KPZ} growth, \href{http://dx.doi.org/10.1023/B:JOSS.0000019810.21828.fc}{\textit{J.~Stat. Phys.}} \textbf{115} (2004),
 255--279, \href{http://arxiv.org/abs/cond-mat/0212519}{cond-mat/0212519}.

\bibitem{Qua11}
Quastel J., Introduction to {KPZ}, in Current Developments in Mathematics, Int.
 Press, Somerville, MA, 2012, 125--194.

\bibitem{QR12}
Quastel J., Remenik D., Local behavior and hitting probabilities of the
 {$\text{Airy}_1$} process, \href{http://dx.doi.org/10.1007/s00440-012-0466-8}{\textit{Probab. Theory Related Fields}}
 \textbf{157} (2013), 605--634, \href{http://arxiv.org/abs/1201.4709}{arXiv:1201.4709}.

\bibitem{QR13}
Quastel J., Remenik D., Airy processes and variational problems, in Topics in
 Percolative and Disordered Systems, \href{http://dx.doi.org/10.1007/978-1-4939-0339-9_5}{\textit{Springer Proc. Math. Stat.}},
 Vol.~69, Springer, New York, 2014, 121--171, \href{http://arxiv.org/abs/1301.0750}{arXiv:1301.0750}.

\bibitem{QS15}
Quastel J., Spohn H., The one-dimensional {KPZ} equation and its universality
 class, \href{http://dx.doi.org/10.1007/s10955-015-1250-9}{\textit{J.~Stat. Phys.}} \textbf{160} (2015), 965--984,
 \href{http://arxiv.org/abs/1503.06185}{arXiv:1503.06185}.

\bibitem{Sas05}
Sasamoto T., Spatial correlations of the 1{D} {KPZ} surface on a f\/lat
 substrate, \href{http://dx.doi.org/10.1088/0305-4470/38/33/L01}{\textit{J.~Phys.~A: Math. Gen.}} \textbf{38} (2005), L549--L556,
 \href{http://arxiv.org/abs/cond-mat/0504417}{cond-mat/0504417}.

\bibitem{Sch12}
Schehr G., Extremes of {$N$} vicious walkers for large {$N$}: application to
 the directed polymer and {KPZ} interfaces, \href{http://dx.doi.org/10.1007/s10955-012-0593-8}{\textit{J.~Stat. Phys.}}
 \textbf{149} (2012), 385--410, \href{http://arxiv.org/abs/1203.1658}{arXiv:1203.1658}.

\bibitem{Sin05}
Singha S.B., Persistence of surface f\/luctuations in radially growing surface,
 \href{http://dx.doi.org/10.1088/1742-5468/2005/08/P08006}{\textit{J.~Stat. Mech. Theory Exp.}} \textbf{2005} (2005), P08006, 17~pages.

\bibitem{Tak12}
Takeuchi K.A., Statistics of circular interface f\/luctuations in an of\/f-lattice
 {E}den model, \href{http://dx.doi.org/10.1088/1742-5468/2012/05/P05007}{\textit{J.~Stat. Mech. Theory Exp.}} \textbf{2012} (2012),
 P05007, 17~pages, \href{http://arxiv.org/abs/1203.2483}{arXiv:1203.2483}.

\bibitem{Tak13}
Takeuchi K.A., Crossover from growing to stationary interfaces in the
 {K}ardar--{P}arisi--{Z}hang class, \href{http://dx.doi.org/10.1103/PhysRevLett.110.210604}{\textit{Phys. Rev. Lett.}} \textbf{110}
 (2013), 210604, 5~pages, \href{http://arxiv.org/abs/1301.5081}{arXiv:1301.5081}.

\bibitem{TS12}
Takeuchi K.A., Sano M., Evidence for geometry-dependent universal f\/luctuations
 of the {K}ardar--{P}arisi--{Z}hang interfaces in liquid-crystal turbulence,
 \href{http://dx.doi.org/10.1007/s10955-012-0503-0}{\textit{J.~Stat. Phys.}} \textbf{147} (2012), 853--890, \href{http://arxiv.org/abs/1203.2530}{arXiv:1203.2530}.

\end{thebibliography}
\end{document}